\documentclass[preprint]{article}

\usepackage{neurips_2026}

\usepackage[utf8]{inputenc} 
\usepackage[T1]{fontenc}   
\usepackage{hyperref}      
\usepackage{url}       
\usepackage{booktabs}      
\usepackage{amsfonts}  
\usepackage{nicefrac}   
\usepackage{microtype}  
\usepackage{xcolor}      
\usepackage{amsmath, amsfonts,amssymb}
\usepackage{graphicx}
\usepackage{subcaption}
\usepackage{tikz}
\usepackage{tikz-3dplot}
\usepackage{tablefootnote}
\usepackage{algorithm}
\usepackage{algpseudocode} 
\usepackage{graphicx}
\usepackage{mathtools}
\usepackage{amsthm}

\usepackage{mathtools}
\DeclarePairedDelimiter{\ceil}{\lceil}{\rceil}

\newtheorem{appendixprop}{Proposition}
\numberwithin{appendixprop}{subsection}
\newtheorem{appendixdef}{Definition}
\numberwithin{appendixdef}{section}

\title{Secure Seed-Based Multi-bit Watermarking for Diffusion Models from First Principles}

\author{%
  Enoal Gesny \\
  Inria\\
  Rennes, France \\
  \texttt{enoal.gesny@inria.fr} \\
  \And
  Eva Giboulot \\
  Inria \\
  Rennes, France \\
}

\begin{document}

\newcommand{\teddy}[1]{\textcolor{red}{#1}}
\newcommand{\evag}[1]{\textcolor{purple}{#1}}
\newcommand{\enoal}[1]{\textcolor{blue}{#1}}
\newcommand\diffmod[1]{\diffmodOP\left({#1}\right)}
\def \Hzero{\mathcal{H}_0}
\def \Hone{\mathcal{H}_1}
\def\un{\mathbbm{1}}
\def\VAE{\mathrm{VAE}}
\def\ie{\textit{i.e.}}
\def\sign{\mathrm{sign}}
\def\PFA{\mathrm{P}_{\mathrm{FA}}}
\def\PD{\mathrm{P}_{\mathrm{D}}}

\newcommand{\method}{\textcolor{blue}{METHOD }}

\def\VSS{\texttt{V-SSig}}
\def\VVS{\texttt{V-VS}}
\def\VTM{\texttt{V-TM}}

\def\GSS{\texttt{G-SSig}}
\def\GVS{\texttt{G-VS}}
\def\GTM{\texttt{G-TM}}

\def\GVSS{\texttt{GV-SSig}}
\def\GVVS{\texttt{GV-VS}}
\def\GVTM{\texttt{GV-TM}}

\def\SS{\texttt{SSig}}
\def\VS{\texttt{VS}}
\def\TM{\texttt{TM}}

\def\bx{\mathbf{x}}
\def\bz{\mathbf{z}}
\def\bk{\mathbf{k}}
\def\bv{\mathbf{v}}

\def\bU{\mathbf{U}}

\newtheorem{proposition}{Proposition}
\newtheorem{property}{Property}
\newtheorem{corollary}{Corollary}
\newtheorem{definition}{Definition}
\def\method{SSB~}

\def\keyset{\mathcal{K}}
\def\E{\mathbb{E}}
\def\P{\mathbb{P}}
\def\V{\mathbb{V}}
\def\partition{\mathcal{P}}

\def\bx{\mathbf{x}}
\def\bX{\mathbf{X}}
\def\bY{\mathbf{Y}}
\def\bZ{\mathbf{Z}}
\def\bxatk{\mathbf{x}_{\mathrm{atk}}}
\def\bxwm{\mathbf{x}_{\mathrm{wm}}}
\def\bm{\mathbf{m}}
\def\bc{\mathbf{c}}
\def\bu{\mathbf{u}}
\def\bU{\mathbf{U}}

\def\DKL{\mathrm{D}_{\mathrm{KL}}}

\def\beps{\boldsymbol{\epsilon}}
\def\bmu{\boldsymbol{\mu}}
\def\bsigma{\boldsymbol{\sigma}}
\def\bSigma{\boldsymbol{\Sigma}}

\def\DKL{\mathrm{D}_{\mathrm{KL}}}


\usetikzlibrary{arrows.meta, calc, fit, positioning, backgrounds}
 
\definecolor{blueBox}{RGB}{210,228,252}
\definecolor{blueBorder}{RGB}{70,130,210}
\definecolor{blueTitle}{RGB}{50,110,190}
 
\definecolor{redBox}{RGB}{252,215,215}
\definecolor{redBorder}{RGB}{210,80,80}
\definecolor{redTitle}{RGB}{190,50,50}
 
\definecolor{greenBox}{RGB}{215,240,220}
\definecolor{greenBorder}{RGB}{70,165,80}
\definecolor{greenTitle}{RGB}{50,140,60}
 
\definecolor{arrowGray}{RGB}{60,60,60}
 
\newlength{\imgW}\setlength{\imgW}{1.6cm}
\newlength{\imgH}\setlength{\imgH}{1.6cm}
 
\tikzset{
  myarrow/.style={->, >=Stealth, thick, dashed, arrowGray,
                  shorten >=2pt, shorten <=2pt},
  solidarrow/.style={->, >=Stealth, thick, arrowGray,
                     shorten >=2pt, shorten <=2pt}
}

\maketitle

\begin{abstract}
    The rapid emergence of generative image models has led to the development of specialized watermarking techniques, particularly in-generation methods such as seed-based embedding. 
    However, current evaluations in this area remain largely empirical, making them heavily reliant on the specific model architectures used for generation and inversion. This prevents any clear conclusion on the performance of any method, especially regarding security, for which a rigorous definition is lacking. 
    Against this approach, we argue that the effectiveness of a watermarking scheme should be established purely through a thorough theoretical analysis. This is enabled by decoupling the model-dependent part from the actual decision mechanism of the watermarking system.
    Using this decoupling, we introduce a formal evaluation framework based on security, robustness, and fidelity. This allows precise comparisons between watermarking systems through a characteristic surface representing the trade-off between these three quantities, independent of any generative model.
    Based on this framework, we propose \method, a novel watermarking method that generalizes previous seed-based methods by allowing to reach any security-robustness-fidelity regime on its characteristic surface. This work opens the door to the design of modern watermarking systems with theoretical guarantees that do not necessitate any costly empirical evaluations.
\end{abstract}

\section{Introduction}\label{sec:intro}

Watermarking involves embedding an invisible signal within the content, which can be retrieved by a detector using a key. This signal is classified as either multi-bit, in which the aim is to recover a binary message, or zero-bit, in which case the detector only decides if the content is watermarked or not, with guarantees in terms of probability of false-alarm. As a specialized data-hiding approach, the goal of a watermarking system is to embed a signal robust to content alterations while preserving content quality with guarantees of security under a given threat model.

Traditional (post-hoc) watermarking methods~\cite{fernandez_video_2024, furon_broken_2008, bui_trustmark_2023} embed the signal directly into the image, which limits their applicability to private APIs. To address the challenges posed by image generation, novel approaches propose to embed the watermark during the generation process. 
This can be performed during the diffusion~\cite{gesny2026guidance}, or through a fine-tuned VAE's decoder~\cite{fernandez_stable_2023}. 

A third approach, which we call \textbf{seed-based}, modifies the seed distribution to incorporate the watermark and subsequently inverts the diffusion process to retrieve it. Interestingly, this approach is also not amenable to open-source models since neither the diffusion model nor the VAE can be modified to enforce the use of a given seed-distribution. However, it possesses the unique property that no distortion is introduced in an already existing latent or image: only the seed distribution is distorted. The promise is thus to allow extreme robustness by finding a distribution that leads to images with similar quality despite being potentially far from the original Gaussian distribution of the seed. 

The approach was first introduced in the zero-bit setting by Tree-Rings~\cite{wen_tree-ring_2023}, which demonstrated that a large distortion of the seed distribution does not necessarily lead to low-quality images. 
Indeed, Tree-Rings is empirically robust to every valuemetric image processing such as JPEG compression. Robustness to geometric operations is currently an open-problem due to the fact that the latent space depends on the size of the input image. Note that the original scheme contained numerous mistakes which were subsequently fixed in~\cite{ci_ringid_2024} and~\cite{gesny2026guidance}[Appendix I]. 

Gaussian-Shading~\cite{yang2024gaussian} translated the idea to the multi-bit setting. It importantly introduced the idea of using a cryptographic primitive controlled by a secret key to generate the latent seed. This has two major benefits. First, it allows to hide the content of the embedded message. Secondly, the resulting scheme is distortion-free: the original seed distribution is preserved \textit{on average with respect to the key}. As such, if each image generated by Gaussian-Shading was produced with different messages and/or keys, the image distribution should be indistinguishable from the non-watermarked model.

All these properties: high image quality, high empirical robustness, and cryptographic security, can be enticing to the practitioner, but we argue that they are fundamentally flawed in their definitions.

Quality and empirical robustness are fundamentally tied to the generative model; this precludes any strong conclusion regarding the performance of a given watermarking scheme. In other words, only propositions of the type "\textit{Watermarking system A is more robust/qualitative than Watermarking system B against transformation T for diffusion model M with VAE V and scheduler S"} can be made. The practitioner would want to make propositions of the type \textit{"Watermarking system A is more robust/qualitative than Watermarking system B"}, whatever the model or operation used. In Section~\ref{sec:method}, we propose to forego the concept of quality in favor of the concept of fidelity. Similarly, we construct a refined definition of robustness independent of the generative model, based on traditional watermarking approaches. This results in a framework that permits such a rigorous ranking of watermarking systems purely on theoretical grounds.

The question of watermarking security and cryptographic security is a lot more nuanced. It is highly dependent on the threat model agreed upon by practitioners. Nevertheless, we argue in Section~\ref{sec:method} that they are very different notions, following the consensus that emerged for classical watermarking methods~\cite{cox_watermarking_2006, bas_measure_2013}. One can indeed have both, yet cryptographic security alone does not prevent some powerful spoofing attacks from being performed on the watermarking system at basically no cost to the attacker.

After this work meant to clarify the three main desirable properties for seed-based watermarking, we arrive at three main contributions. First, we propose a methodological framework for comparing seed-based watermarking algorithms, focusing on the trade-offs between security, capacity, and fidelity independently of the diffusion model used. Then, we design a new seed-based watermarking based on this framework that generalizes current multi-bit approaches, allowing to achieve a wide variety of security-fidelity-capacity tradeoff. Finally, we provide an empirical validation of the theoretical guarantees of the scheme.

\section{Decoupling the decision from the diffusion}\label{sec:method}

Before starting, it is important to consider what counts as an acceptable watermark system. This requires a clear delimitation of the problem we are trying to solve with watermarking. Consider the following illustrative scenario:

\paragraph{Watermarking attribution scenario} Alice makes accessible an API where users can request for images to be generated. In order to trace the use of her system, she asks Bob, a third-party, to provide her with a secret key $k$ and a seed-based watermarking system $\mathcal{W}$. From the point of view of Bob, the secret key $k$ is now uniquely linked to Alice. She then associates to each user a unique ID $\bm_i$. Each request generates an image, using $\mathcal{W}$, that contains a watermark. If Bob is presented with the key $k$ and a watermarked image generated for user $i$, he should decode the correct message $\bm_i$. 

The most important constraint of this scenario is that Alice \textbf{cannot change her secret key} but \textbf{can use as many messages as she wants}. Contrary to what is often assumed in the current watermarking literature, we argue that it is not trivial to work with many keys when using a multi-bit watermarking system. Indeed, if Alice uses multiple keys, attribution becomes difficult if not impossible: if she wants to extract the message from a given image, how does she determine which key to trust? Recall that Alice does not know from which user the image came. Furthermore, one cannot solve the problem by associating a unique key to each user and choosing the correct one for decoding; this requires an oracle. If such an oracle exists, then watermarking is not necessary since attribution is solved simply by using this oracle. 

The rest of the section is dedicated to developing a rigorous evaluation framework for watermarking systems, aiming to address the \textit{attribution scenario}. For clarity and space considerations, we settle on a mostly informal and intuitive tone in the paper itself, but always link to formal definitions in Appendix~\ref{app:definitions} for important notions. We also advise the reader to peek at Appendix~\ref{app:notation} for a review of notations used in the paper.

\subsection{Watermarking systems}\label{subsec:wm-sys}
A watermarking system $\mathcal{W}$ -- Def.~\ref{def:embedding_mechanism} -- can be decomposed into four main elements: a set of secret keys $\mathcal{K}$, a projection function $f$, a family of embedding mechanisms $(e_{k})_{k\in \keyset}$, and a decision mechanism $d$. The projection function $f:\mathbb{R}^D \rightarrow \mathbb{R}^L$ ($L \leq  D$) transforms a piece of content into a representation more amenable to embedding and detection -- the so-called watermark space. It should be robust to a family of usual content alterations. 
Embedding mechanisms $e_{k}$ -- Def.~\ref{def:embedding_mechanism} -- are functions parameterized by the secret key. Their goal is to push a point in watermark space deep inside a specific region called the decoding region. The watermarked distribution obtained is denoted as $\mathcal{Q}_{(k,\bm)}$.
The decision mechanism $d$ -- Def.~\ref{def:embedding_mechanism} -- maps points of the watermark space to either a score (zero-bit watermarking) or a message (multi-bit watermarking). Our interest in this paper lies solely in multi-bit schemes. 
We will see in Section~\ref{subsec:robustness} that it is useful to separate the redundancy mechanism -- Def.\ref{def:redundancy_mechanism} -- (i.e, the error-correcting code) from the decision mechanism. 
For the rest of the paper, we thus make the distinction between a message $\bm$ of size $M$ and its corresponding representative codeword $\bc$ of size $M^{\prime} > M$.  When using a redundancy mechanism, the watermarking system becomes an error-corrected watermarking system --  Def.~\ref{def:error_corrected_ws} -- which notably changes the watermarking distribution from $\mathcal{Q}_{(k,\bm)}$ to $\mathcal{Q}_{(k,\bc)}$.

One important property of a watermarking system is that it should be difficult to randomly find a piece of content that follows the desired watermarked distribution. We formalize this with the concept of cover distribution, denoted as $\mathcal{P}$ -- Def.~\ref{def:cover_distrib}. Informally, we expect message bits decoded from images sampled from a cover distribution to be independent and not biased in favor of a particular message. Such property is usually enforced by whitening a decoder~\cite{fernandez_stable_2023,gesny2026guidance} or by designing so-called sound detectors from the ground-up~\cite{fernandez_three_2023}.

Finally, seed-based watermarking enforces a specific structure on the cover distribution that will form the basis of the analysis in Section~\ref{sec:analysis}:

\begin{definition}[Seed-based system -- Def.\ref{def:seed-based}]
    A watermarking system $\mathcal{W}$ equipped with cover distribution $\mathcal{P}$ is said to be seed-based if the projection of non-watermarked content in watermark space is distributed as a standard Gaussian. That is we have that: $f(X) \sim \mathcal{N}(\mathbf{0}, \mathbf{I}_L), X \sim \mathcal{P}$.
\end{definition}

We now define the three main quantities of interest for a watermarking system: capacity, fidelity, and watermarking security. We contrast these properties to the more common trio of robustness, quality, and cryptographic security. For clarity, we split each subsection into two parts. The first discusses the current approach in measuring a given quantity, and why we believe it to be unsatisfying. The second formalizes our approach in an abstract and general fashion. 

\subsection{Robustness versus Capacity}\label{subsec:robustness}

All recent works on seed-based studies empirically assess robustness by evaluating the chosen detection statistics under a set of image transformations~\cite{fernandez_stable_2023,fernandez_video_2024, yang2024gaussian, bui_trustmark_2025,gesny2026guidance}. 
We claim that this empirical approach for evaluating robustness suffers from a major problem:

\paragraph{Problem} No definite claim on the ranking of different watermarking designs can be made since the performance depend on every choice of the inverse diffusion pipeline, notably: the choice of diffusion model, the scheduler/solver, the number of inversion steps, the choice of VAE, the presence or absence of prompts at inversion, and in the former case the value of the guidance scale. The inverse diffusion can be improved, for example by the choice of a better solver. However, we claim this can hardly be called an improvement over the watermarking system itself, only that it benefits from a "better" estimation of the seed.

Traditional watermarking sidestepped these issues by evaluating the robustness of a given scheme theoretically. The main idea was to assume that any alteration in pixel space translates to a white Gaussian noise with a given power in watermark space~\cite{chen_quantization_2000, pateux_practical_2003, furon_broken_2008}. Thanks to the properties of the latent space in seed-based watermarking, we can expect this model of robustness to be a good approximation of reality -- something we empirically validate in Section~\ref{sec:experiments}.

Consequently, it would be natural to model the decision mechanism as an additive (potentially not) white Gaussian noise channel  $\mathrm{AWGN}_\sigma$. However, since bit-accuracy is so prevalent in current watermarking evaluations, it is fruitful instead to model the decision mechanism as $M^{\prime}$ independent binary symmetric channels with flip probability $p$ -- noted $\mathrm{BSC}_p$. We discuss this choice in Appendix~\ref{app:channel-model}: the main takeaway is that Shannon's capacity is a good summary of the overall watermarking system performance, which depends on both the robustness $\sigma$ of the projection $f$ and the so-called channel characteristic $p_k(\sigma)$ of the decision mechanism $d$. 

Indeed, notice that the flip probability $p$ depends on two parts of the system: the projection $f$, and the watermarked distribution $\mathcal{Q}_k$. In practice, an alteration in pixel space introduces some noise in watermark space. This noise will, in turn, impact $p$, with the impact depending on $\mathcal{Q}_k$. This is where we perform the decoupling between $f$ and $d$. We measure the \textit{robustness} of $f$ against a transform $t$  by the variance $\sigma^2$ of the noise $t $introduces in watermark space -- Def.~\ref{def:robustness}. We then define the \textit{decision channel characteristic} $p_k(\sigma)$ of $d$ as the mapping between $\sigma$ and $p$ -- see Def.~\ref{def:decision_chanel_characteristic}. Informally, the channel characteristic is the bit-error rate of the system, \textit{before error-correction}, under a white Gaussian perturbation with variance $\sigma^2$.

The combination of both quantities leads to the Shannon capacity, which perfectly summarizes how much noise a watermarking system can resist for a given codeword size.

\begin{proposition}[Capacity of a watermarking system]
Let $t : \mathbb{R}^D \rightarrow \mathbb{R}^D$ be a function. A watermarking system $\mathcal{W}$ with projection $f$ that is $\sigma$-robust to $t$ and with channel characteristic $p_k$ has a Shannon capacity $C_\sigma = 1- h_2(p_k(\sigma))$, where $h_2$ is the binary entropy function. 
\label{prop:capacity}
\end{proposition}
This means that the system cannot communicate the watermark reliably unless the message $\bm$ is coded with a codeword $\bc$ of size $M^{\prime} \geq \ceil{\frac{M}{C_\sigma}}$. 
Except when specifically stated, we assume to work at Shannon's capacity in order to be agnostic to the choice of redundancy mechanism. As such, the decision mechanism is \textit{only} determined by its channel characteristic, which maps a given noise power $\sigma^2$ to a given flip probability $p$. In practice, one has to choose an off-the-shelf code depending on the desired properties. Turbo codes~\cite{berrou_near_1993} and Polar codes~\cite{arikan_channel_2009} allow the use of rates close to Shannon's capacity. Pseudo-random codes (PRC)~\cite{christ_pseudorandom_2024} provide cryptographic security. Gaussian-Shading use of repetition code~\cite{yang2024gaussian}[Section 3.2], though highly suboptimal in every sense, is simple to implement and analyse.

\subsection{Quality versus Fidelity}

Quality is straightforward to define in the case of post-hoc watermarking. It suffices to compare the host image to the watermarked one. In particular, the PSNR is equivalent to the power of the watermark signal in this case. This is not possible in our case.

\paragraph{Problem 1} In-generation watermarking is peculiar because it lacks any original host image. Instead, the standard approach in the literature is to use empirical distance between the distribution of images generated by the watermarked model and a given distribution of un-watermarked images (generated or not)~\cite{gesny2026guidance, wen_tree-ring_2023, gunn_undetectable_2025}. This is usually performed using the FID~\cite{heusel_gans_2017}, which is notoriously unable to describe some important sampling failure modes~\cite{sajjadi_assessing_2018, simon_revisiting_2019}. Notably, the diversity of content generation is difficult to assess using this metric. As an example, for a given key and message, the support of the seed-distribution of Gaussian-Shading is cut in half, leading to an observed lack of variety in generated content~\cite{gunn_undetectable_2025}[Figure 13] but no observed impact on the FID. The choice of baseline distribution is also not trivial if one is concerned with reproducibility and comparability: papers report baseline FID that differ substantially between one another~\cite{fernandez_stable_2023, wen_tree-ring_2023, yang2024gaussian,gesny2026guidance, gunn_undetectable_2025,li_gaussmarker_2025}. Another approach is to use "semantic" distances such as the CLIP score~\cite{hessel_clipscore_2022}, but this is also not satisfactory: a watermarked model that preserves semantics but significantly degrades the aesthetic content is not desirable.

This overall difficulty in empirical evaluation also demonstrates that there is a lack of clarity in what constitutes an acceptable degradation of the generative model.

We argue that the role of the watermark designer is not to guarantee a certain level of generation quality; that is the role of the model provider. The role of the watermark designer is rather to guarantee that the watermarked model \textit{preserves} the properties of the original one. In the case of seed-based watermarking, the seed distribution is what impacts the generation. As such, the watermark designer wants to preserve the original distribution as much as possible. The quantification of the deviation from this original distribution is what we call fidelity.

\paragraph{Problem 2} Under our scenario, the secret key is fixed, and a user is linked to a unique message. Thus, a distortion-free scheme is not sufficient for our purpose: it only guarantees that, on average over the messages, the properties of the non-watermarked model are preserved. A case in point is again Gaussian-Shading: each user is attributed only half of the available sampling space for seeds. On average over the users, all of the sampling space is used. Yet, an unlucky user might only be able to access seeds that lead to "low-quality" images. This is not desirable and points to the need to \textit{preserve the original distribution as much as possible for all keys and messages}.

Starting from the fundamental property of seed-based systems in Def.~\ref{def:seed-based}, it is natural to measure fidelity as an f-divergence between the standard multivariate Gaussian distribution and a watermarked distribution $\mathcal{Q}_{(k,\bc)}$. For this paper, we settle on the Kullback-Leibler (KL) divergence, though it will require our analysis to use some asymptotic arguments in the next section.

\begin{definition}[Fidelity of Seed-based watermarking -- Def.\ref{def:fidelity}]
     A watermarked distribution $\mathcal{Q}_{(k,\bc)}$ of a seed-based watermarking system is said to be $\zeta$-faithful iff $\DKL(\mathcal{N}(\mathbf{0},I_L) || \mathcal{Q}_{(k,\bc)}) \leq \zeta$. We call $\zeta$ the relative fidelity loss with respect to the Gaussian distribution.
\end{definition}

\subsection{Cryptographic versus Watermarking Security}\label{sec:wm_secu}

\paragraph{Problem: Watermarking is not cryptography} The question of security for modern watermarking is slowly emerging as trivial vulnerabilities are found in state-of-the-art post-hoc schemes due to the lack of secret keys~\cite{bas_ai_2025, tarhini_neural_2026}. Seed-based watermark systems have tackled this question under the cryptographic angle, both in Gaussian-Shading and the more recent applications of PRC~\cite{christ_pseudorandom_2024,gunn_undetectable_2025}. The idea conveyed by these works is that a watermarking system is secure if an attacker cannot infer anything of the embedded message from the watermarked seed. The conflation between this type of cryptographic security and watermarking proper is not new and dates back to the early 2000s. The problem was eloquently settled by Cox et al. in~\cite{cox_watermarking_2006}, which demonstrates that cryptography and security solve different problems. Maybe the key difference between watermarking and cryptography outlined in this work is the following ``the secret carrier does not need to be exactly disclosed in order to break the watermarking system''~\cite{cox_watermarking_2006}[Section 4.1]''. We point out that both Gaussian-Shading and PRC do not even hide the encrypted message within a secret carrier! It is readily available to any attacker by simply taking the sign of a watermarked seed. 

Once again, it is useful to illustrate this problem with a scenario that will serve as our threat model:

\paragraph{Threat model: Spoofing attack} Eve is a user of Alice's API. She is a friend of Kerckhoffs, an employee of Alice who disclosed to her the full watermarking system, $\mathcal{W}$ except for the secret key $k$. She thus knows that each user is given a unique message. She would like to impersonate Camille, another user from whom she has stolen $N_o$ generated images. Since she was disclosed the watermarking system, she can retrieve the watermarked seed from each image. Furthermore, she knows that all Camille's images contain their unique ID $\bm_{\text{Camille}}$. Note that in order to make her system more secure, Alice prevents a user from using the same seed twice. Hence, Eve cannot simply reuse Camille's seed to generate new images. 

Under this threat model, both Gaussian-Shading and PRC are highly insecure. Even though Alice does not know the exact value of $\bm_{\text{Camille}}$, she has access to a codeword $\bc_{\text{Camille}}$ linked to Camille, which is sufficient for her: encryption does not bring any security in this situation. Indeed, since neither scheme makes the carrier secret, she can simply retrieve the encrypted message, sample a new seed from it, and generate a new image with it. This underlines the fundamental difference of watermarking security: \textit{the codeword carrier needs to be secret}.

We evaluate the security of a watermarking system by the availability of an estimator $\psi$ to retrieve the secret key $k$. We convey this approach through the notion of security ratio. Informally, it measures the number of watermarking samples necessary for Eve to mount an attack that is better than a naive brute-force:

\begin{definition}[Security ratio -- Def.\ref{def:security_ratio}]
    A $M$-bit error-corrected watermarking system $\mathcal{W}$ with codeword size $M^{\prime}$ is $\eta$-secure against an estimator $\psi$ if it requires at least $N_o = \eta L$ watermarked observations to estimate a fixed key $k\in \keyset$ better than randomly guessing. That is, for any codeword $\bc \in \{0,1\}^{M^{\prime}}$:
    \begin{align}
    \eta L &= \min \{N \mid  \mathbb{P}[ d\left(e_{\psi_k(N)}(\bc), k\right) = \bc] > \frac{1}{2^{M^{\prime}}}, N \in \mathbb{N}^+\}
    \end{align}
    where $\psi_k(N)$ is an estimation of the key $k$ based on $N$ i.i.d. samples $Z \sim \mathcal{Q}_{(k,\bc)}$. We call $\eta$ the security ratio of $\mathcal{W}$ against $\psi$. If $\eta = +\infty$, we say that $\mathcal{W}$ is perfectly secure against $\psi$.
\end{definition}
Since both Gaussian-Shading and PRC do not use a secret key to conceal their codewords, it only necessitates a single sample to mount an attack, hence $\eta=1/L$ for both of them.

\section{A general lattice-based construction}\label{sec:ssb}

We propose Secure Seed-Based (SSB), a  multi-bit watermarking scheme heavily inspired by older schemes from the classic watermarking literature based on Voronoï modulation~\cite{chen_quantization_2000,chaumont_tcq_2011} (also see Chapter 9 in~\cite{zamir_lattice_2014}). Compared to these works, our approach is substantially simplified: we use two nested one-dimensional lattices in order to securely transmit a binary codeword. Despite the simplicity of the approach, this decision mechanism can operate under a wide range of capacity-fidelity-security trade-off, generalizing and largely outperforming current approaches. A surprising result is that perfectly secure transmission in the sense of Def.\ref{def:security_ratio} is often possible even under high noise variance $\sigma^2$ -- see Appendix~\ref{app:perfect-security}.

Our focus here is solely on the \textit{embedding function} and \textit{decision mechanism}. We do not try to optimize the estimation of the original seed by improving the projection function (i.e. the inverse diffusion process). Similarly, recall that we are agnostic to the choice of error-correcting code; the only quantity of interest is the decision channel characteristic $p(\sigma)$ from Def.~\ref{def:decision_chanel_characteristic}.
Figure~\ref{fig:diagram_method} illustrates the functioning of the method.

\begin{figure}[t]
    \centering
    \resizebox{0.9\linewidth}{!}{%
    \input{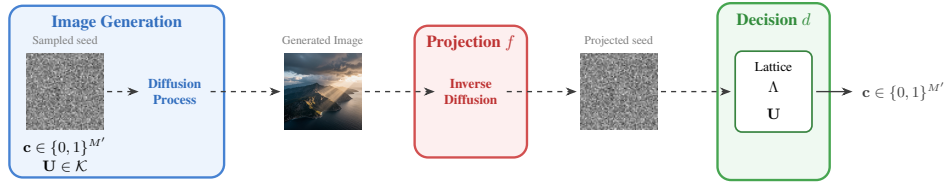}
    }
    \caption{This is the diagram of \method proposed. The unitary matrix $\bU$ is used as the key. The projection $f$ is the inverse diffusion from the image in pixel space to an approximation of the seed in the latent space. The decision mechanism $d$ is the use of the lattice function $\Lambda_\Delta$.}
    \label{fig:diagram_method}
\end{figure}

\subsection{Embedding and decoding}

Let $\bc$ be the binary codeword to hide.
We define the key set $\keyset_u$ as the set of matrices in $\mathcal{M}^{L\times M^\prime}$ which have their columns summing to one -- in particular if $M^\prime = L$, the matrix is unitary. Let  $\bU \in \keyset_u$.
Finally, let the alternating lattice function $\Lambda_{\Delta} : \mathbb{R}^{M^\prime} \rightarrow \{0, 1\}^{M^\prime}$ be defined as:
\begin{equation}
    \Lambda_\Delta(\bx) =\frac{ (-1)^{\lfloor \frac{\bx}{\Delta} \rfloor }+ 1}{2},
    \label{eq:lattice_alternate}
\end{equation}
with $\Delta$ the size of a cell. We call $\Lambda_\Delta$ the \textbf{coarse lattice}.
Note that when $\Delta \to +\infty$, the lattice function tends to the sign function: this is the decision mechanism of Gaussian-Shading.

Our goal is to sample a vector  $\bz_u \in \mathbb{R}^{M^\prime}$ such that the lattice outputs the correct codeword: $\Lambda_\Delta(\bz_u) = \bc$.
Once this is done, we have to project this vector into the original latent space $\mathbb{R}^{L}$.
In order to improve the system's capacity, we use a second, nested, lattice function $\Lambda_\delta$ defined in the same way as in Eq.~\eqref{eq:lattice_alternate} but with $\delta \leq \Delta$. We call any such $\Lambda_\delta$ a \textbf{fine lattice}.
An illustration of these embedding functions is proposed in Figure~\ref{fig:sampling}. Algorithm~\ref{alg:algo} details the sampling of $\bz_u$.

\begin{figure}[htbp]
    \centering
    \begin{minipage}{0.5\textwidth}
        \centering
        \input{algo/algo1}
        \captionof{algorithm}{Seed sampling algorithm for SSB. In practice, we truncate $k$ between $-10$ and $10$ for numerical evaluations.}
        \label{alg:algo}
    \end{minipage}
    \hfill 
    \begin{minipage}{0.39\textwidth}
        \centering
        \includegraphics[width=\linewidth]{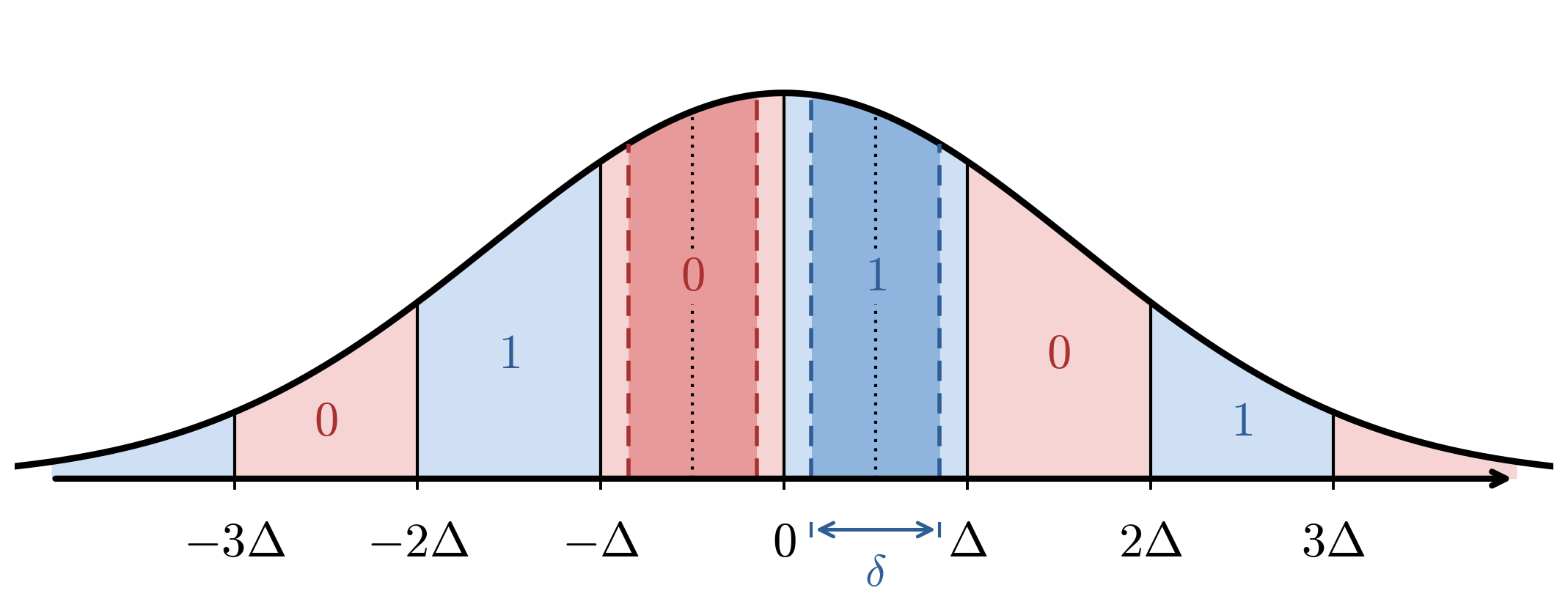}
        \caption{Diagram of the nested lattices $\Lambda_{\Delta}$ and $\Lambda_{\delta}$. The coarse lattice cells are separated by full lines wheres fine cells are colored with dotted lines.}
        \label{fig:sampling}
    \end{minipage}
\end{figure}

Now we still need to transform $\bz_u$ into a valid seed for the diffusion process. Any seed $\bz \in \mathbb{R}^L$ can be decomposed as $\bz = (\mathbf{I}_{L} - \bU\bU^\top) \bz + \bU\bU^\top \bz$. In order to obtain a valid seed $\bz$ from $\bz_u$, one can thus sample $\bz^{\prime}$ as a realization of a random variable $Z^{\prime} \sim \mathcal{N}(\mathbf{0}, \mathbf{I}_L)$ and compute $\bz$ as: 
\begin{equation}\label{eq:latent-space-projection}
 \bz = (\mathbf{I}_L  - \bU\bU^\top) \bz^{\prime} +  \bU\bz_u.
\end{equation}

For decoding, let $\bx \in \mathbb{R}^D$ be an image watermarked with the codeword $\bc$. We compute an approximation of the original latent $\hat{\bz} = f(\bx)$. The codeword is then decoded simply by applying the coarse lattice function on $\bU^\top \hat{\bz}$ giving $\Lambda_\Delta(\bU^\top \hat{\bz}) = \hat{\bc}$. By construction, if the projection does not introduce any noise, $\bU^\top \hat{\bz} = \bU^\top \bz = \bz_u$ and consequently we retrieve the correct codeword: $\bc=\hat{\bc}$.

\subsection{Theoretical Analysis}\label{sec:analysis}

A direct consequence of Alg.~\ref{alg:algo} is that each watermarked element $(z_u)_i$ is distributed as an infinite mixture of truncated Gaussian distributions. Using the notation in Alg.\ref{alg:algo} (or Appendix~\ref{app:notation}), and writing the probability to sample in the $k$-th lattice cell as $P_k = 2(\Phi(b_k) - \Phi(a_k))$, we have that: $(z_u)_i = \mathrm{sign}(c_i)\sum_{k \in \mathbb{Z}}P_k Z_k$ with $Z_k \sim  \mathrm{Trunc}\mathcal{N}_{[\alpha_k, \beta_k]}\left(0,1\right)$. Since all watermarked elements are independent and follow the same distribution up to a sign, we note the unsigned expectation of $(z_u)_i$ as $\mu_{(\Delta,\delta)} $ and its variance as $\sigma_{(\Delta,\delta)} ^2$ -- see Proposition~\ref{prop:watermarked-dist-moments} for an exact formula of these two moments.

\paragraph{Capacity} The flip probability of an SSB decision mechanism is given by the probability that an element $(z_u)_i$ corrupted by some Gaussian noise with variance $\sigma^2$ leaves the set of "correct" coarse lattice cells. Since the distribution of $(z_u)_i$ is known, the channel characteristic $p_{(\Delta,\delta)}(\sigma)$ of SSB is immediately given by Proposition~\ref{prop:channel-carac}. We plot the resulting capacity $C_\sigma(\Delta,\delta)$ computed numerically using Proposition~\ref{prop:capacity} in Figure~\ref{fig:ssb-fullcharac}.

\begin{proposition}[\method channel characteristic]\label{prop:channel-carac}
Let $(\Delta,\delta)$ define a \method system. Then the probability that a codeword bit is flipped due to a transform against which $f$ is $\sigma$-robust is: 
\begin{equation}
      p_{(\Delta,\delta)}(\sigma) =  1 - \left( 2 \sum_{j=\mathbb{Z}}P_j \sum_{k=\mathbb{Z}} \frac{\int_{\alpha_j}^{\beta_j} \phi(p)\left[\Phi\left(\frac{b_k - p}{\sigma} \right) -\Phi\left(\frac{a_k - p}{\sigma} \right)\right] dp}{\Phi(\beta_j) - \Phi(\alpha_j)} \right)
    \label{eq:robustness_multibit}
\end{equation}

\end{proposition}

\begin{figure}[t]
    \centering
    \begin{subfigure}[t]{0.33\textwidth}
    \includegraphics[width=0.9\linewidth]{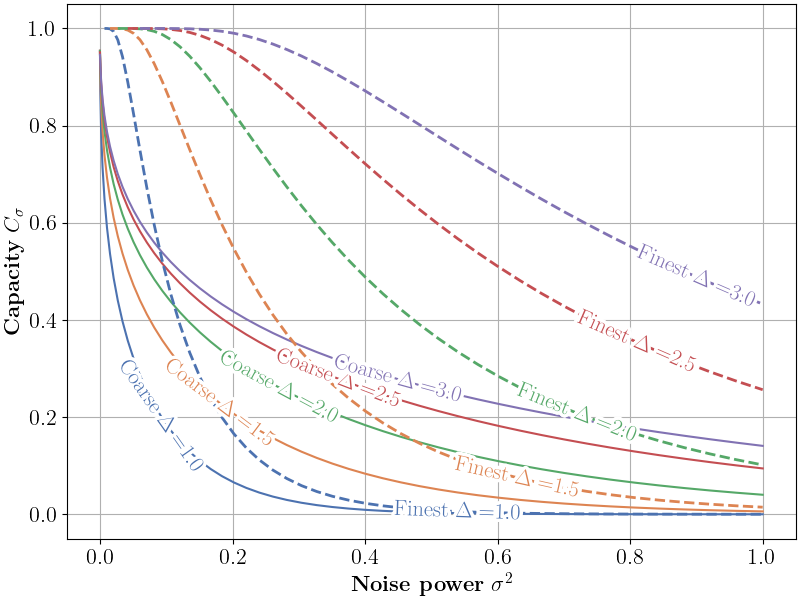}
    \label{fig:ssb_capacity}
    \end{subfigure}%
    \begin{subfigure}[t]{0.33\textwidth}
    \includegraphics[width=0.9\linewidth]{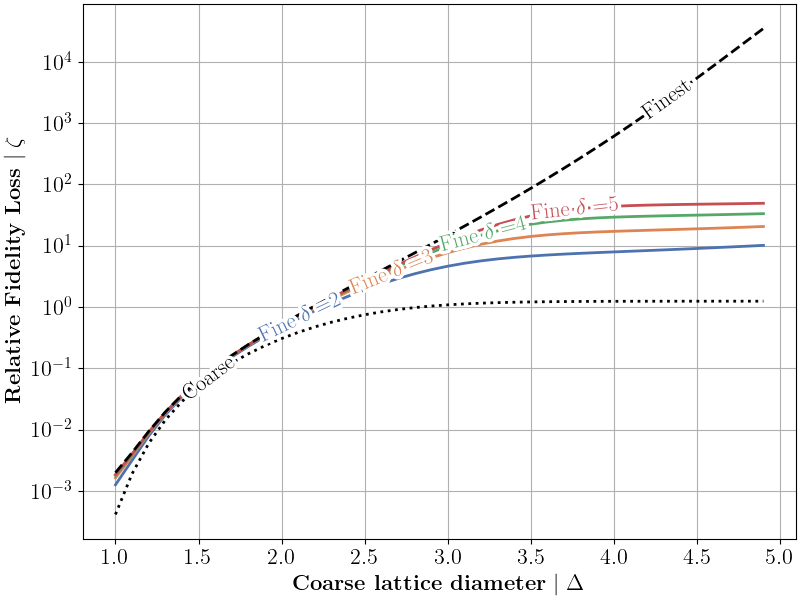}
    \label{fig:ssb_fidelity}
    \end{subfigure}%
    \begin{subfigure}[t]{0.33\textwidth}
    \includegraphics[width=1\linewidth]{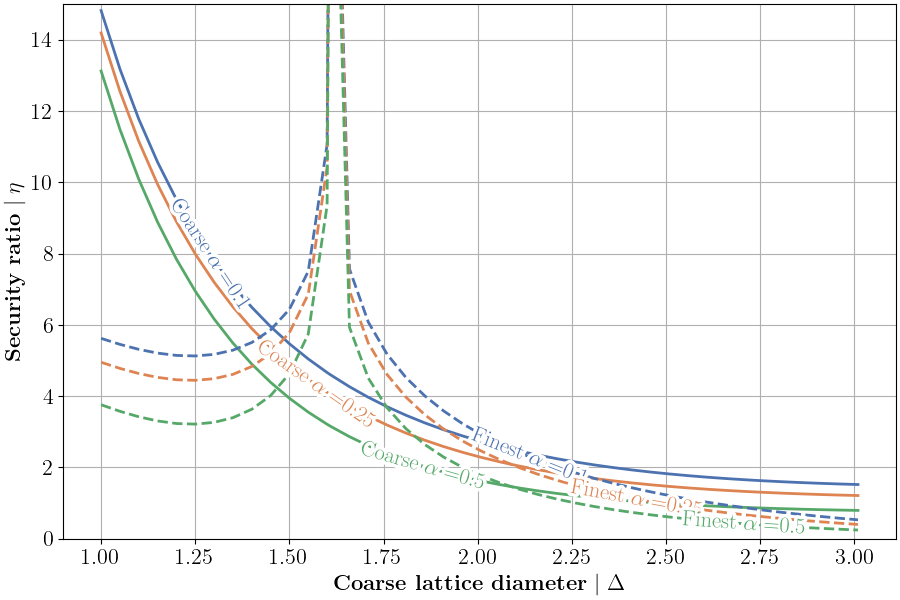}
    \label{fig:ssb_security}
    \end{subfigure}
    \caption{Watermarking characteristic of $(\Delta,\delta)$-\method, showing the capacity, fidelity and security as a function of the lattice parameters as well as of the codeword size relative to the latent size $\alpha = M^{\prime}/L$ }\label{fig:ssb-fullcharac}

\end{figure}

\paragraph{Fidelity} Recall that under our watermarking scenario we \textit{must} assume both the codeword $\bc$ and the secret rotation $\bU$ to be fixed. The support of the watermarked distributions induced by any lattice embedding is consequently on a strict subset of $\mathbb{R}^L$. The KL-divergence thus diverges with respect to the Gaussian distribution. 
However, $\bU$ is applied on independent random variables with finite mean and variance, leading to a "Gaussianization" of the distribution of the watermarked seed. Consequently, we propose to compute the fidelity on the asymptotic distribution of the watermarked seed obtained by application of the Lyapunov central-limit theorem.

\begin{proposition}[\method Asymptotic relative fidelity]\label{prop:ssb-fidelity}
    Any $(\Delta,\delta)$-SSB system with codeword size $M^\prime$ is $M^{\prime}\zeta(\Delta, \delta)$-faithful with $\zeta(\Delta, \delta) \xrightarrow[]{L\rightarrow \infty} \frac{1}{2}\left[\left(1 + \mu_{(\Delta,\delta)} ^2 \right)\sigma_{(\Delta,\delta)} ^{-2} + \log \sigma_{(\Delta,\delta)} ^2 - 1\right]$ if $\bU$ is chosen such that Lyapunov condition holds.
\end{proposition}

\paragraph{Security: Principal Component Analysis (PCA) Attack}
In our watermarking scenario, an adversary, Eve, aims to estimate the secret key $\bU$ from $N_o$ watermarked observations. The eigenvalues of the covariance matrix of cover latents are the identity. However, the \method embedding mechanism modifies these eigenvalues. As such, Eve's best strategy relies on identifying deviation from the expected distribution of the eigenvalues of the empirical covariance estimated with the watermarked samples. Her best estimator is limited by the behavior of the Marchenko-Pastur distribution of the eigenvalues of empirical covariance matrices~\cite{bianchi_performance_2011, vallet_improved_2012}. We refer to Appendix~\ref{app:mp-threat-model} as well as \cite{furon_fast_2013} for a primer on this type of attack. From Def.~\ref{def:security_ratio}, we consider security broken as soon as a single component of $\bU$ is revealed by the estimator. The number of samples necessary to do so is given in Proposition~\ref{proposition:ssb-sec}.

\begin{proposition}[\method Watermarking security]\label{proposition:ssb-sec}
 Let $\alpha = \frac{M^\prime}{L}$. Any $(\Delta,\delta)$-SSB is $\eta$-secure against a PCA estimator where $\eta = \left( \frac{1- \sqrt{\alpha}\sigma_{\Delta,\delta}}{1 - \sigma_{\Delta,\delta}} \right)^2$.
\end{proposition}

\section{Experimental validation}\label{sec:experiments}

\begin{figure}[t]
    \centering%
    \begin{subfigure}[t]{0.25\textwidth}
    \includegraphics[width=1\linewidth]{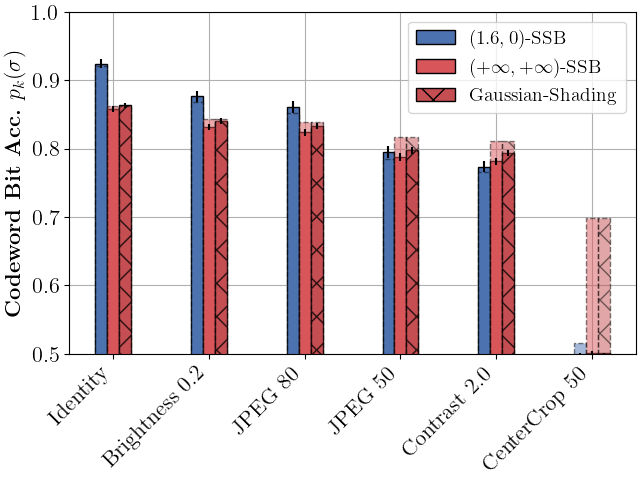}
    \label{fig:bit-acc-valid}
    \end{subfigure}%
    \begin{subfigure}[t]{0.25\textwidth}
    \includegraphics[width=1\linewidth]{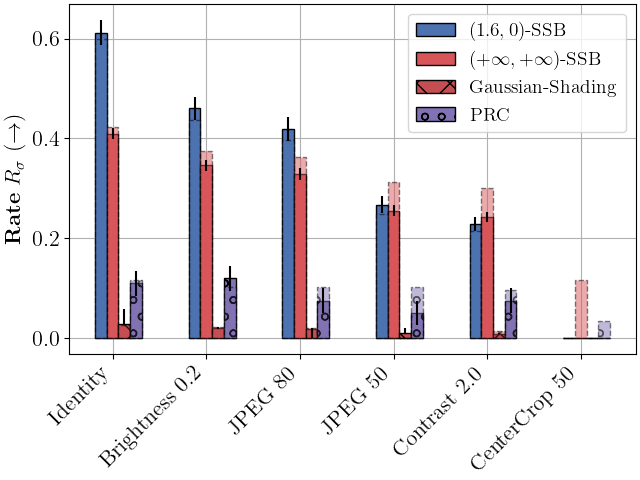}
    \label{fig:rate-valid}
    \end{subfigure}%
    \begin{subfigure}[t]{0.25\textwidth}
    \includegraphics[width=1\linewidth]{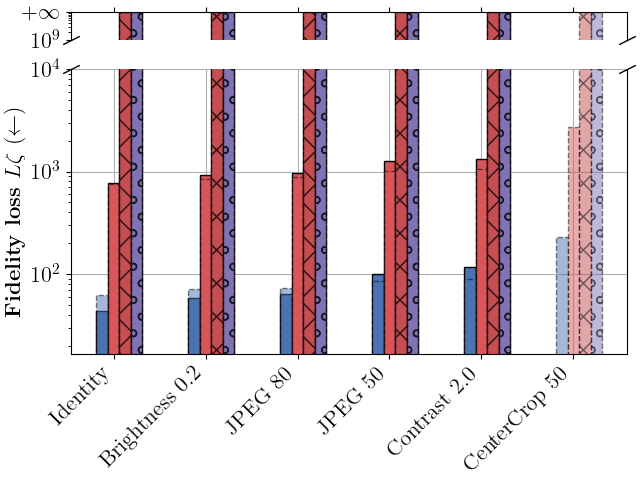}
    \label{fig:rate-valid}
    \end{subfigure}%
    \begin{subfigure}[t]{0.25\textwidth}
    \includegraphics[width=1\linewidth]{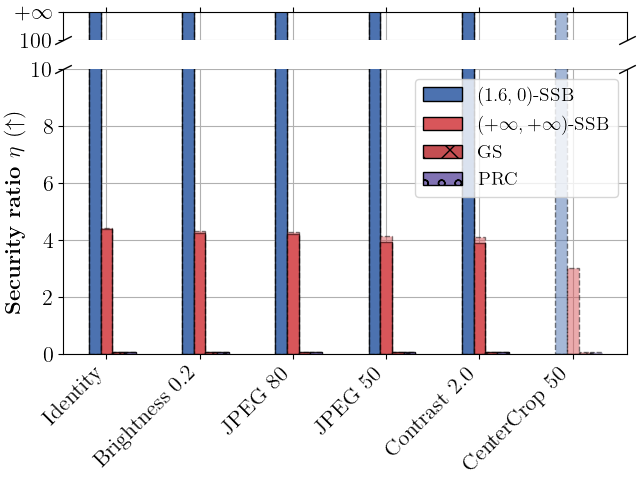}
    \label{fig:rate-valid}
    \end{subfigure}%
    \caption{Empirical validation of the model defined in Section~\ref{sec:method} for current seed-based approaches -- Gaussian-Shading and PRC -- as well as our method \method. Transparent boxes correspond to theoretical values, opaque to empirical ones. We display $95\%$ asymptotic confidence intervals.}\label{fig:emp-valid}
\end{figure}

\subsection{Robustness of the inverse diffusion}\label{subsec:experiments-transform}

We have defined robustness in Def.\ref{def:robustness} as the mapping between a certain image processing operation and an equivalent noise variance $\sigma^2$ introduced in watermarked space. For seed-based systems, robustness depends on two main \textit{empirical} components: the diffusion model used for inversion and the choice of input distribution $\mathcal{F}$. For the latter, one should, in all rigor, use watermarked images as input distributions in order to match the hypotheses of Proposition~\ref{prop:capacity}. However, for most image processing operations of interest, choosing a cover distribution -- i.e., non-watermarked images -- should be an extremely good proxy and allows one to provide an evaluation agnostic to the choice of watermarking system. We show in Section~\ref{subsec:emp-validation} that our assumption is indeed empirically valid. 
We measure the robustness of current popular flow-matching diffusion models against common image processing operations in Table~\ref{tab:robustness}. We chose Sana~\cite{xie_sana_2024}, Z-Image Turbo~\cite{team_z-image_2025}, and Qwen~\cite{wu_qwen-image_2025} due to their high popularity on Huggingface at the time of writing, as well as their difference in architecture. Sana's VAE provides very high-compression rates, Z-Image Turbo is a distilled model, and Qwen is a very large model necessitating many (50) diffusion steps compared to the two other models\footnote{We deliberately chose \textit{not} to use Stable-Diffusion 2 due to replicability concerns, as it was removed from Huggingface}.

Our robustness evaluation protocol is as follow. We generated $N=100$ images $(\bx)_N$  from a collection of $N$ seeds $(\bz)_N$ using a given model and prompts from Huggingface:\url{Gustavosta/Stable-Diffusion-Prompts}. The projection $f$ is then defined as the standard inverse flow matching procedure using the Euler solver, for the same number of diffusion steps as during generation. We then estimate the variance of the introduced noise as the empirical variance of the residuals $f(\bx) -\bx$ averaged over latent channels. 

\begin{table}[tb]
\centering
\resizebox{0.94\linewidth}{!}{
\begin{tabular}{l c c || c  c  c  c c  c}
\toprule
\textbf{LDM/Steps} & \textbf{Image size} $D$ & \textbf{ Latent size} $L$ & \textit{Identity} & \textit{Brightness} 0.2 & 
\textit{Contrast} 2.0 & \textit{JPEG} QF80 & \textit{JPEG} QF50   & \textit{Center Crop} 50\% \\
\midrule

\textbf{Sana}/ 25 & $512^2$& $8192$ & 0.21 & 0.29 & 0.46 & 0.31 & 0.42 & 1.94
\\
\textbf{Z-image} / 9 & $1024^2$& $262144$ &0.35 & 0.45 & 0.66 & 0.9 & 1.08 & 1.51
\\
\textbf{Qwen} / 50&  $512^2$& $65536$ & 0.34 & 0.44 & 0.78 & 0.78 & 1.09 & 2.09
 \\

\bottomrule
\end{tabular}
}

\caption{Equivalent white Gaussian noise variance $\sigma^2$ introduced in latent space for different common image operations -- see Def.~\ref{def:robustness}. The \textit{Identity} transforms corresponds to the noise introduced by the inverse diffusion process itself ($f$ in Def.\ref{def:m-b_watermarking_system}). All diffusion models are used with the standard configuration provided one their Huggingface page -- see Appendix~\ref{app:licences}.}
\label{tab:robustness}
\end{table}

\subsection{Comparing empirical and theoretical performances}\label{subsec:emp-validation}
We now validate our theoretical results to the empirical performance of \method, Gaussian-Shading, and PRC. We choose to test two regimes for \method $(\Delta, \delta)$ = $(+\infty, +\infty)$ and $(1.6, 0)$.  The first regime uses a decision mechanism identical to Gaussian-Shading's -- the $\mathrm{sign}$ function. The second corresponds to the perfect security regime of \method ($\eta = +\infty$) -- see Appendix~\ref{app:perfect-security}. For each watermarking system, we generate $500$ images with the same prompts as in the previous sections using the Sana diffusion model. We first compute the empirical bit-accuracy of the \textit{codeword}, which corresponds to the channel characteristic $p(\sigma)$ in Def.\ref{def:decision_chanel_characteristic}. From it we compute the optimal rate $R_\sigma \leq C_\sigma$ -- see Def.\ref{prop:capacity}  -- for which we are guaranteed to decode the correct message with negligible probability of decoding error $\mathrm{P}_e$. For \method, we set the rate to Shannon's capacity at the empirical $p(\sigma)$. Gaussian-Shading uses a repetition code; we thus compute the exact best rate that leads to a probability of decoding error $\mathrm{P}_e \leq 10^{-6}$ for a given $p(\sigma)$. Finally, PRC uses a belief-propagation code, BP+OSD~\cite{roffe_decoding_2020}, which is not compatible with our $\mathrm{BSC}_p$ model since it uses soft-decisions. As such, we directly report the best theoretical and empirical rate for PRC. Since computing the best empirical rate would be far too costly for $\mathrm{P}_e \leq 10^{-6}$, we report it for $\mathrm{P}_e \leq 10^{-2}$, giving it a slight advantage. Fidelity and security are computed in turn using these rates. We report all results in Figure~\ref{fig:emp-valid}. Note that none of the current seed-based systems can resist cropping at $50\%$.

\section{Conclusion}\label{sec:conclusion}

Modern watermarking methods lack a shared language. 
Methods are proposed and evaluated on different axes, and compared empirically in ways that conflate distinct design choices -- the inverse diffusion model, the decision mechanism, and the error-correcting codes. 
This conflation makes it impossible to know whether a method is genuinely better, or simply makes a different trade-off. This work introduced a framework that allows fair theoretical comparison. By decoupling these components and anchoring the analysis in capacity, fidelity, and watermarking security, we bring multi-bit seed-based watermarking to align with classic watermarking methods. 
We also introduced a concrete implementation of this framework, \method, a watermarking system that achieves a wide range of regimes along these three axes. 
Because each component is decoupled, each can be improved independently, and future work can do so without redesigning the whole system. 
We hope this work shifts the focus away from ad-hoc optimizations based on empirical evaluations, and towards design with strong theoretical guarantees. Several directions follow naturally from this work. 
First, it is clear that the robustness of the projection function can be improved, for example by a better model of variance on each latent channel. Second, our analysis assumes a BSC$_p$: understanding the gap with an AWGN$_\sigma$ remains to be studied.
Finally, extending this framework to 0-bit seed-based watermarking is a natural (but not trivial!) next step. 

{
    \small
    \bibliographystyle{plain}
    \bibliography{main}
}

\appendix

\section{Notation}\label{app:notation}

\paragraph{Spaces}
\begin{itemize}
    \item \textbf{Pixel space}: in $\mathbb{R}^D$ with observations denoted $\bx$.
    \item \textbf{Latent space}: in $\mathbb{R}^L$ with observations denoted $\bz$.
    \item \textbf{Watermark space}: in $\mathbb{R}^{M^\prime}$ with observations denoted $\bz_u$.
\end{itemize}

\paragraph{Distributions} We use the standard notation for the standard Gaussian p.d.f $\phi$ and its c.d.f $\Phi$. We denote the Truncated Gaussian with location $\mu$ and scale parameter $\sigma$ and boundaries $(a,b)$ as $\mathrm{Trunc}\mathcal{N}_{[a,b]}\left(\mu,\sigma^2\right)$. Other distributions are usually referred to with calligraphic letters $(\mathcal{N})$.

\paragraph{Secure Seed Based watermarking system} We refer to a $(\Delta, \delta)$-\method system as an error-corrected watermarking system (Def.\ref{def:error_corrected_ws}), where:
\begin{itemize}
    \item The set of secret keys is given by the set of $L\times M^{\prime}$ matrices with columns summing to $1$. An element of this set is denoted as $\mathbf{U}$.
    \item The projection function $f$ is any function from $\mathbb{R}^D$ to $\mathbb{R}^L$.
    \item The embedding functions are given by Alg.~\ref{alg:algo} followed by Eq.~\eqref{eq:latent-space-projection}.
    \item The decision mechanism is simply $\Lambda_\Delta$ as defined in Eq~\ref{eq:lattice_alternate}.
\end{itemize}
We assume the redundancy mechanism to work at Shannon's capacity. A message is denoted as $\bm \in \{0,1\}^M$ and its representative codeword as $\bc \in \{0,1\}^{M^\prime}$. For convenience, we use a slightly modified version of the sign function defined as:
\begin{align}
    \mathrm{sign}(x) = \begin{cases}
    1 &\text{ if } x > 0\\
    -1 &\text{ else }
    \end{cases}
\end{align}
Importantly, note that $\mathrm{sign}(0) = -1$.

\paragraph{Lattice cells}

\begin{itemize}
    \item $(a_k)_{k\in\mathbb{Z}} = 2k\Delta$ and  $(b_k)_{k\in\mathbb{Z}} = a_k + \Delta$ are the boundaries of the coarse ($\Delta$) lattice cells.
    \item $(\alpha_k)_{k\in\mathbb{Z}} = a_k + \frac{\Delta-\delta}{2}$ and $(\beta_k)_{k\in\mathbb{Z}} = \alpha_k + \delta$ are the boundaries of the fine ($\delta$) lattice cells.
    \item $P_k \triangleq 2(\Phi(b_k) - \Phi(a_k))$.
\end{itemize}

\section{Formal Definitions}\label{app:definitions}

\subsection{Watermarking systems}
This appendix aims to make precise the informal definitions found in Section~\ref{sec:method}. A first reading of the paper can safely skip these details. However, note that the soundness of the evaluation framework and the corresponding proofs explicitly rely on the definitions found herein. For the reader's convenience, we also collect definitions already provided in the main text. 
\begin{appendixdef}[Embedding mechanism]
An embedding mechanism $e$ is a random variable that maps a message $\bm \in \{0,1\}^M$ to a subset of the watermarked space $\mathcal{D}_\bm \subset \mathbb{R}^L$. The induced probability distribution is denoted as $\mathcal{Q}_\bm$ and called the watermarked distribution.
\label{def:embedding_mechanism}
\end{appendixdef}
Note that $\mathcal{D}_\bm$ is to be understood as a \textit{subset} of the decoding region of the decoder for message $\bm$: a good embedding mechanism favors mappings deep inside the correct decoding region in order to maximize robustness.

\begin{appendixdef}[Multi-bit watermarking System]
    A $M$-bit watermarking system $\mathcal{W}$ is a quadruplet $(\mathcal{K}, f, (e_{k})_{k\in \keyset}, d)$ where $\mathcal{K}$ is the set of secret keys, $f: \mathbb{R}^D \to \mathbb{R}^L$ the projection function, $(e_k)_{k\in \keyset}$ a family of embedding mechanisms indexed by secret keys and $d: \mathbb{R}^L \times \mathcal{K} \rightarrow \{0,1\}^M$ the decision mechanism. 
    \label{def:m-b_watermarking_system}
\end{appendixdef}

\begin{appendixdef}[Redundancy mechanism]
    A redundancy mechanism is composed of an encoder and a decoder $(c_{\mathrm{enc}},c_{\mathrm{dec}})$. The encoder $c_{\mathrm{enc}}$ is a bijection between messages and a subset $\mathcal{C} \in \{0,1\}^{M^{\prime}}$ called the codebook. 
    For each codeword $\bc$ in the codebook, the decoder $c_{\mathrm{dec}}$ defines an equivalence class $[\bc] = \{\bv \in \{0,1\}^{M^{\prime}} : c_{\mathrm{dec}}(\bv) = \bc \}$ and where $c_{\mathrm{dec}} : \{0,1\}^{M^{\prime}}$. We will abuse notation and make no distinction between the equivalence class $[\bc]$ and its representative codeword $\bc$.
    \label{def:redundancy_mechanism}
\end{appendixdef}

\begin{appendixdef}[Error-corrected watermarking system]
A $M$-bit watermarking system $\mathcal{W}$ equipped with a redundancy mechanism $(c_{\mathrm{enc}},c_{\mathrm{dec}})$ replaces its embedding functions $(e_k)_{k\in\keyset}$ by $e_k^{(c)} \triangleq e_k \circ c_{\mathrm{enc}}$. The induced watermarked distribution are denoted as $\mathcal{Q}_{(k,\bc)}$. Its decision mechanism is replaced by a function $d^{(c)} \triangleq  c_{\mathrm{enc}}^{-1} \circ c_{\mathrm{dec}} \circ d$ where $d: \mathbb{R}^L \times \mathcal{K} \rightarrow \{0,1\}^{M^{\prime}}$. We call $d^{(c)}$ the error-corrected decision mechanism; we still call $d$ the decision mechanism.
\label{def:error_corrected_ws}
\end{appendixdef}

\begin{appendixdef}[Cover distribution]
    Let $\mathcal{P}$  be a probability distribution with support in $\mathbb{R}^D$. It is said to be a cover distribution for a $M$-bit watermarking system $\mathcal{W}$ iff, for all secret keys $k \in \keyset$:
    \begin{equation}
        d\left(f(X),k\right) \sim \mathcal{B}^M\left(\frac{1}{2}\right),
    \end{equation}
    where $X \sim \mathcal{P}$ and $\mathcal{B}^M$ is a $M$-dimensional Bernoulli distribution.
    \label{def:cover_distrib}
\end{appendixdef}

\begin{appendixdef}[Seed-based system]
    A watermarking system equipped with cover distribution $\mathcal{P}$ is said to be seed-based if the projection of non-watermarked content in watermark space is distributed as a standard Gaussian. That is we have that: $f(X) \sim \mathcal{N}(\mathbf{0}, \mathbf{I}_L), X \sim \mathcal{P}$.
    \label{def:seed-based}
\end{appendixdef} 
\subsection{Watermarking characteristic properties}

\begin{appendixdef}[Robustness]
    A projection function $f$ is said to be $\sigma$-robust to a function $t : \mathbb{R}^D \rightarrow \mathbb{R}^D$ under an input distribution $\mathcal{F}$ (with support in $\mathbb{R}^D$) iff :
    \begin{equation}
        \V[(f \circ t)(X)] \leq \sigma^2, X \sim \mathcal{F}.
    \end{equation}
    \label{def:robustness}
\end{appendixdef}

\begin{appendixdef}[Decision channel characteristic]
		Let $\mathcal{W}$ be a $M$-bit error-corrected watermarking system with codeword size $M^{\prime}$. The channel characteristic $p_k$ of $\mathcal{W}$ for a key $k\in \keyset$ is a mapping $p_k: \mathbb{R}^+ \rightarrow [0,1]$ defined as: 
		\begin{align}
			p_k(\sigma) = \frac{1}{M^{\prime}2^M}\sum_{\bc \in \mathcal{C}}  \E_{Z \sim Q_{(k,\bc)}}\left[\E_{\epsilon_\sigma }\left[\mathrm{ham}(d\left(Z + \epsilon_\sigma, k\right), d\left(Z, k\right)) \right]\right],
        \end{align}
        where $\epsilon_\sigma \sim \mathcal{N}(\mathbf{0}, \sigma^2 I_L)$ and $\mathrm{ham}$ is the Hamming distance.
        \label{def:decision_chanel_characteristic}
\end{appendixdef}

\begin{appendixdef}[Fidelity of Seed-based watermarking]
     A watermarked distribution $\mathcal{Q}_{(k,\bc)}$ of a seed-based watermarking system is said to be $\zeta$-faithful iff:
    \begin{equation}
        \DKL(\mathcal{N}(\mathbf{0},I_L) || \mathcal{Q}_{(k,\bc)}) \leq \zeta.
    \end{equation}
    We call $\zeta$ the relative fidelity loss with respect to the Gaussian distribution.
    \label{def:fidelity}
\end{appendixdef}

\begin{appendixdef}[Security ratio]
    A $M$-bit error-corrected watermarking system $\mathcal{W}$ with codeword size $M^{\prime}$ is $\eta$-secure against an estimator $\psi$ if it requires at least $N_o = \eta L$ watermarked observations to estimate a fixed key $k\in \keyset$ better than randomly guessing. That is, for any codeword $\bc \in \{0,1\}^{M^{\prime}}$:
    \begin{align}
    \eta L &= \min \{N \mid  \mathbb{P}[ d\left(e_{\psi_k(N)}(\bc), k\right) = \bc] > \frac{1}{2^{M^{\prime}}}, N \in \mathbb{N}^+\}
    \end{align}
    where $\psi_k(N)$ is an estimation of the key $k$ based on $N$ i.i.d. samples $Z \sim \mathcal{Q}_{(k,\bc)}$. We call $\eta$ the security ratio of $\mathcal{W}$ against $\psi$. If $\eta = +\infty$, we say that $\mathcal{W}$ is perfectly secure against $\psi$.
    \label{def:security_ratio}
\end{appendixdef}


\section{Further discussion about the watermarking system modeling}

\subsection{Choice of channel model}\label{app:channel-model}

We decided to model the decision mechanism as a binary symmetric channel $\mathrm{BSC}_p$ instead of an additive white Gaussian noise $\mathrm{AWGN}_\sigma$. 

This choice does lead to a small loss in capacity: the "hard decision" prevents some capacity-achieving codes from reaching their full potential~\cite{costello_channel_2007}[Section C]. Nevertheless, this ultimately facilitates the analysis and comparisons to other multi-bit watermarking schemes. It also separates the choice of decision mechanism and coding mechanism, providing great flexibility in the choice of the latter depending on the desired properties needed for the watermarking system (e.g. maximal capacity versus cryptographic security).

Shannon's noisy channel coding theorem~\cite{polyanskiy_information_2025}[Theorem 19.9], provides the sufficient (asymptotic) rate at which one must work to achieve an arbitrarily small probability of error\footnote{We are aware of the recent non-asymptotic bounds on the capacity~\cite{polyanskiy_channel_2010}, but the Shannon regime conveniently allows an analysis independent of probability of decoding error.}. The strong converse theorem~\cite{polyanskiy_information_2025}[Section 22.1 and Theorem 22.1] shows this rate to be necessary for the $\mathrm{BSC}_p$ channel. Importantly, the error exponent converges exponentially fast to $0$ or $1$. This makes the asymptotic study of the channel straightforward since the actual value of the probability of decoding error does not matter: it can be as small (though not zero) as one desires for rates below capacity; the probability of error is essentially one for rates above capacity. The main quantity of interest thus becomes the flip probability $p$, which fully determines the capacity.

Now, the flip probability $p$ depends on two parts of the system: the robustness of the projection $f$, and the watermarked distribution $\mathcal{Q}_k$. In practice, an alteration in pixel space introduces some noise in watermark space. This noise will, in turn, impact $p$, with the impact depending on $\mathcal{Q}_k$. This is where we perform the decoupling between $f$ and $d$. We measure the \textit{robustness} of $f$ against a transform $t$ by the power $\sigma^2$ of the noise said transform introduces in watermark space -- see Definition~\ref{def:robustness} in the Appendix\ref{app:definitions}. We finally define the \textit{decision channel characteristic} $p_k$ of $d$ as the mapping between $\sigma^2$ and $p$ -- see Definition~\ref{def:decision_chanel_characteristic} in the Appendix~\ref{app:definitions}. 
The combination of both quantities leads to the Shannon capacity, which perfectly summarizes how much noise a watermarking system can resist for a given codeword size.

\subsection{Security, PCA attack and the Marchenko-Pastur distribution}\label{app:mp-threat-model}

This section relies on the threat model introduced in Section~\ref{sec:analysis}, instantiated for the \method method described in Section~\ref{sec:analysis}.

\paragraph{Setup}
Eve observes $N_o$ watermarked samples in order to recover the secret key $\bU$. 
In this context, Eve's best strategy is the Principal Components Analysis (PCA): Eve computes the empirical covariance matrix of the observations and performs an eigendecomposition, hoping that the directions of $\bU$ appear as distinguishable principal components. 

\paragraph{PCA Attack}
Given $N_o$ observations, Eve computes the empirical covariance matrix $\hat{\Sigma}_\bz = \frac{1}{N_o}\sum_i \bz_i \bz_i^\top$, and looks at its spectrum.
In the absence of a watermark, the observations are i.i.d with a covariance matrix $I_L$. 
Random matrix theory predicts that the empirical eigenvalues follow the Marchenko-Pastur (MP) distribution $\mathrm{MP}(\frac{L}{N_o}, 1)$. 
If the watermark leads to a covariance shift, so some eigenvalues depart from $1$, the corresponding directions stand out as outliers relative to the expected MP distribution. 
Eve can use them to estimate $\bU$.

\paragraph{Example analysis}
Figure~\ref{fig:eigenvalues} gives illustrative examples of eigenvalue distributions. 
First, Figure~\ref{fig:eigenvalues_nw} illustrates the baseline, when no watermark is present in $\bz$. 
The empirical spectrum closely matches the MP distribution. 
When a watermark seed sampling is performed, the true covariance may not be the identity anymore. 
In the case illustrated by Figure~\ref{fig:eigenvalues_gs}, the theoretical eigenvalues are $\lambda_0 = 1$ for $L - M^\prime$ dimensions and $\lambda_1 = 1 - \frac{2}{\pi}$ for $M^\prime$ dimensions.
For a sufficiently large $N_o$ ($10L$ in the figure), Eve observes two distinct MP distributions, each centered around one of these eigenvalues. 
By identifying the subspace associated with $\lambda_1$, Eve obtains an estimation of $\bU$, making the scheme insecure for $N_o$ observations. 
To satisfy the security ratio definition -- Def.~\ref{def:security_ratio} -- the support of the watermarked empirical spectrum must be contained within the support of the non-watermarked MP distribution. 
This condition ensures that no outlier eigenvalue can betray a direction of $\bU$. 
This condition is achieved for all $N_o$ when $\lambda_o = \lambda_1 = 1$. 
Figure~\ref{fig:eigenvalues_perfect_secu} illustrates this perfect-security regime for the $(1.6,\,0)$-SSB 
scheme detailed in Appendix~\ref{app:perfect-security}. 
The empirical spectrum is indistinguishable from the unwatermarked case, so Eve gains no information about the secret key $\bU$.

\begin{figure}[t]
    \centering
    \begin{subfigure}[t]{0.33\textwidth}
    \includegraphics[width=1\linewidth]{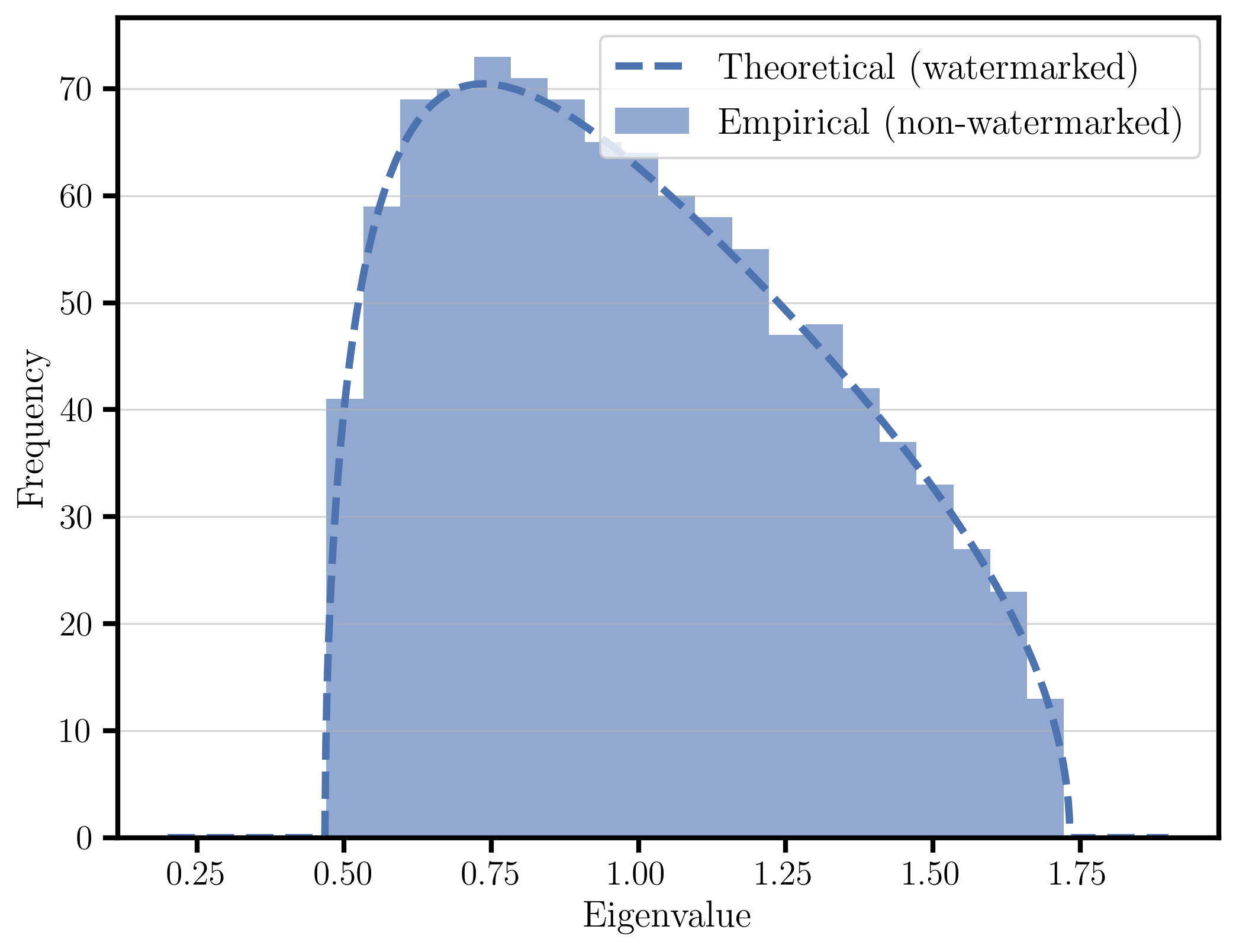}
    \caption{No watermark}
    \label{fig:eigenvalues_nw}
    \end{subfigure}%
    \begin{subfigure}[t]{0.33\textwidth}
    \includegraphics[width=1\linewidth]{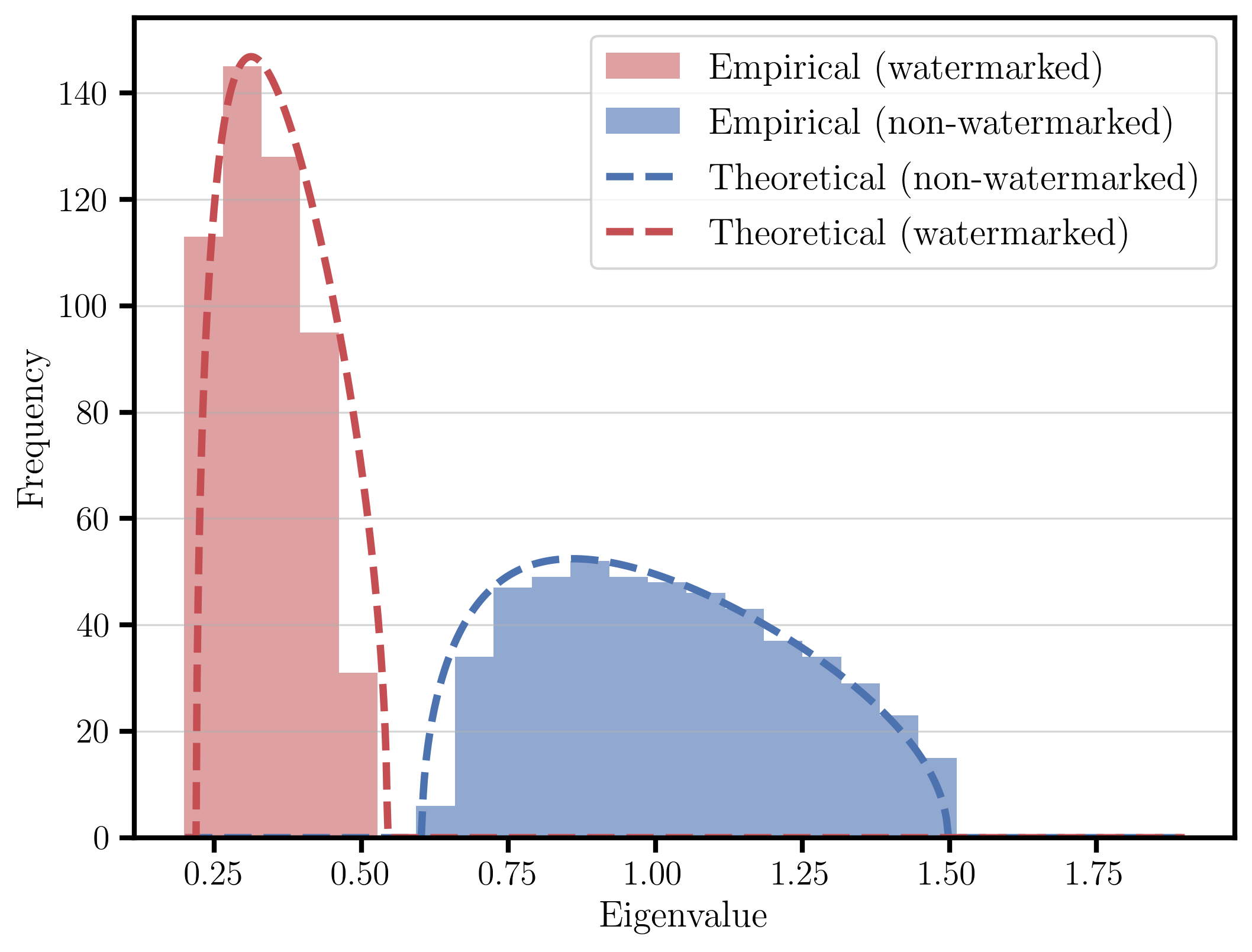}
    \caption{$(+\infty, +\infty)$-SSB}
    \label{fig:eigenvalues_gs}
    \end{subfigure}%
    \begin{subfigure}[t]{0.33\textwidth}
    \includegraphics[width=1\linewidth]{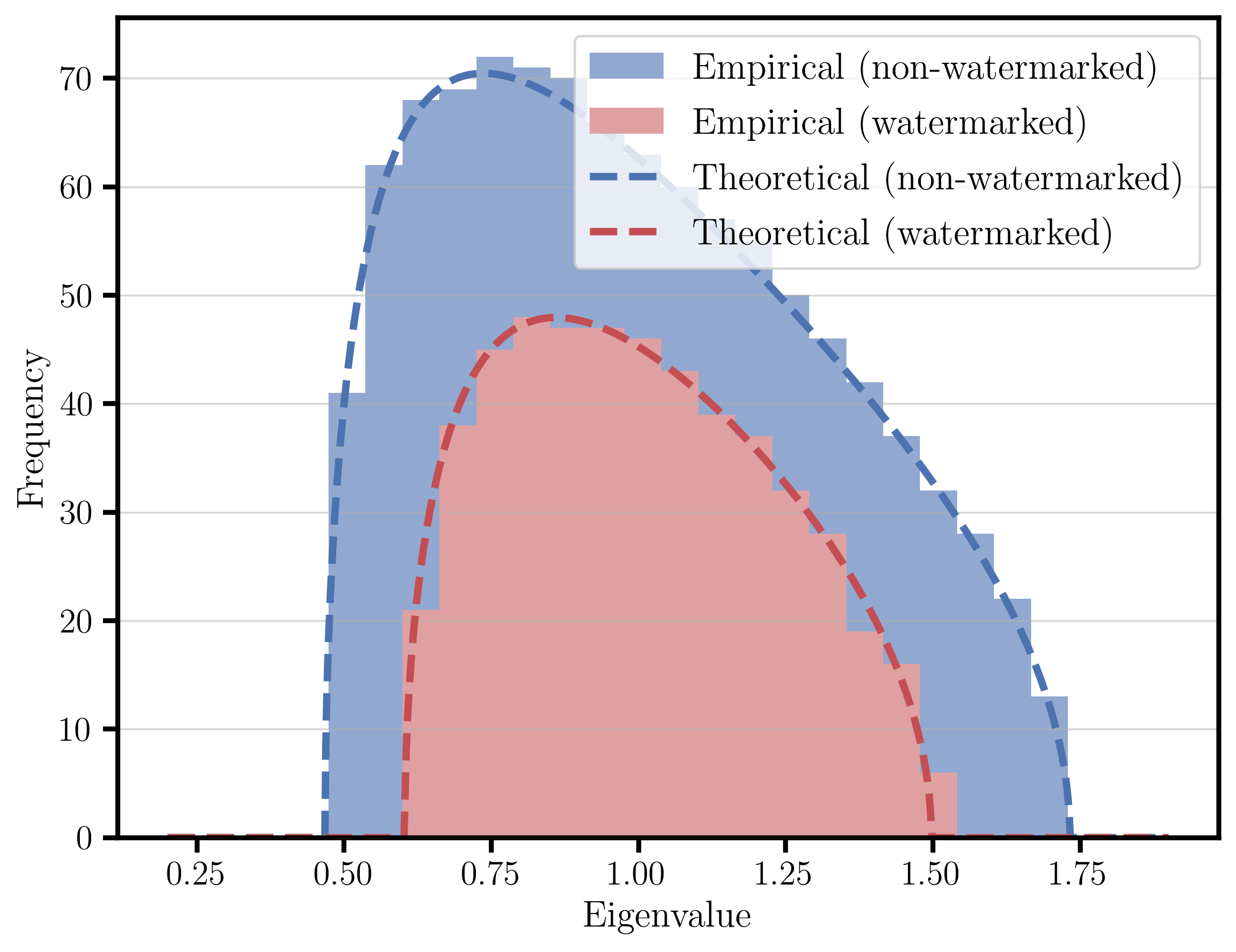}
    \caption{$(1.6, 0)$-SSB}
    \label{fig:eigenvalues_perfect_secu}
    \end{subfigure}
    \caption{Comparison between empirical and theoretical eigenvalues of the covariance matrix for different sampling of the seed. The empirical eigenvalues are computed over $N_o = 10L$ samples with $L=512$ and $M^\prime = 256$.}\label{fig:eigenvalues}

\end{figure}

\section{Proofs}\label{app:proofs_section}

\subsection{Fundamental properties of \method}

All the subsequent proofs rest on the fact that each element in watermarked elements in \textit{watermark space} $(z_u)_i$ are independent and distributed as an infinite mixture of truncated standard Gaussian random variables. For a given codeword $\bc$ we always have:
\begin{align}
    (z_u)_i &= \mathrm{sign}(c_i)\sum_{k \in \mathbb{Z}}P_k Z_k,\\
    Z_k &\sim  \mathrm{Trunc}\mathcal{N}_{[\alpha_k, \beta_k]}\left(0,1\right).
\end{align}
This is explicitly constructed in Alg.\ref{alg:algo} and as such does not need further proof.

On the other hand, the watermarked distribution in the latent space $\mathcal{Q}_{(\bU, \bc)}$ is more difficult to describe as a weighted sum of infinite mixtures. However, thanks to the nature $\bU$, it is simple to obtain the asymptotic watermarked distribution. First, we quantify the first two moments in the watermarked space.

\begin{appendixprop}[Embedding distribution moments]\label{prop:watermarked-dist-moments}
    Let $(\Delta,\delta)$ define a \method system. 
    Then  we have that, for any $\bz_u \sim \mathcal{Q}_{(\bU, \bc)}$, 
    \begin{align}
        \mathbb{E}[\bz_u] \triangleq \bmu_{(\Delta,\delta)} &= \mathrm{sign}(\bc)\sum_{k=-\infty}^{+\infty} P_k\frac{\phi(\alpha_k) - \phi(\beta_k)}{\Phi(\beta_k) - \Phi(\alpha_k)},\\
        \mathrm{Cov}[\bz_u] \triangleq \bSigma_{(\Delta,\delta)}  &= \mathbf{I}_L\left(\sum_{k=-\infty}^{+\infty} P_k \left(1-\frac{\beta_k\phi(\beta_k) - \alpha\phi(\alpha_k)}{\Phi(\beta_k) - \Phi(\alpha_k)}\right) - \bmu_u^2 \right).
    \end{align}
\end{appendixprop}

\begin{proof}
See Section~\ref{app:proof-MP-sec}. 
\end{proof}

Observing that all watermarked elements are independent and follow the same distribution up to a sign, we note the unsigned expectation of $(z_u)_i$ as $\mu_{(\Delta,\delta)} $ and its variance as $\sigma_{(\Delta,\delta)} ^2$.

From this, one can derive the asymptotic distribution of the watermarked latent, which is Gaussian if the weights of $\bU$ are chosen suitably. We don't provide an explicit solution for guaranteeing that $\bU$ is "good" in this sense since
in practice, one only has to ensure the entries of $\bU$ are close to a uniform distribution to ensure Lyapunov's condition holds and that enough entries per row are non-zero. This can be done by sampling a matrix in $\mathbb{R}^{L\times M^{\prime}}$ where each entry is sampled from a uniform distribution. Each column is then normalized such that they sum to $1$. Importantly, note that the matrix where each column contains a single $1$ is not a good matrix for \method.

\begin{appendixprop}[Asymptotic watermarked distribution]\label{app:asymp-dist}
    Let $(\Delta,\delta)$ define a \method system with fixed key $\bU$ and codeword $\bc$.  We have that the watermarked distribution $\mathcal{Q}_{(\bU, \bc)}$ converges as $L\rightarrow\infty$ in distribution to a multivariate Gaussian distribution $\mathcal{N}(\bU^T\bmu_{(\Delta,\delta)},\bU^T\bSigma_{(\Delta,\delta)}\bU)$ if $\bU$ is chosen such that Lyapunov condition holds.
\end{appendixprop}
\begin{proof}
    We provide a sketch of the proof using a slightly weaker result where the infinite sums in Proposition~\ref{prop:watermarked-dist-moments} are truncated such that $|k| \leq \kappa$. For any choice of $\kappa$, the moments are then obviously finite since $\Phi(x)$ is strictly increasing and bounded away from zero and from one. Consequently, $\forall |k| < \kappa$, we always  have $\Phi(\beta_k) - \Phi(\alpha_k) > 0$.
    The proposition assumes Lyapunov condition to hold, which concludes the proof. 
\end{proof}

\subsection{Proposition~\ref{prop:capacity}: Capacity of a watermarking system}
\begin{appendixprop}[Capacity of a watermarking system]
Let $t : \mathbb{R}^D \rightarrow \mathbb{R}^D$ be a function. A watermarking system $\mathcal{W}$ with projection $f$ that is $\sigma$-robust to $t$ and with channel characteristic $p_k$ has a Shannon capacity $C_\sigma = 1- h_2(p_k(\sigma))$, where $h_2$ is the binary entropy function. 
\label{prop:app-capacity}
\end{appendixprop}
\begin{proof}
    Since the decision mechanism is modeled as a binary symmetric channel with flip probability $p_k(\sigma)$, this the classic application of Shannon's noisy channel coding theorem~\cite{polyanskiy_information_2025}[Theorem 19.9] to a $\mathrm{BSC}_p$ channel.
\end{proof}

\subsection{Proposition~\ref{prop:channel-carac} : \method Channel characteristic}\label{app:sec_proof}
\begin{appendixprop}[\method channel characteristic]\label{prop:app-channel-carac}
Let $(\Delta,\delta)$ define a \method system. Then the probability that a codeword bit is flipped due to a transform against which $f$ is $\sigma$-robust is: 
\begin{equation}
      p_{(\Delta,\delta)}(\sigma) =  1 - \left( 2 \sum_{j=\mathbb{Z}}P_j \sum_{k=\mathbb{Z}} \frac{\int_{\alpha_j}^{\beta_j} \phi(p)\left[\Phi\left(\frac{b_k - p}{\sigma} \right) -\Phi\left(\frac{a_k - p}{\sigma} \right)\right] dp}{\Phi(\beta_j) - \Phi(\alpha_j)} \right)
    \label{eq:robustness_multibit}
\end{equation}

\end{appendixprop}

\begin{proof}
We denote by $1 - p_{(\Delta,\delta)}(\sigma)$ the probability that a codeword bit is not flipped due to a transform against which $f$ is $\sigma$-robust.

\begin{align}
    1 - p_{(\Delta,\delta)}(\sigma) &= \sum_{j=\mathbb{Z}} \frac{1}{Z_j}\int_{\alpha_j}^{\beta_j} \left[2\sum_{k=\mathbb{Z}} \int_{a_k}^{b_k} \phi\left(\frac{x - p}{\sigma}\right) dx \right] \phi(p) dp, \\
\end{align}
with $Z_j = \Phi(\beta_i) - \Phi(\alpha_i)$ the normalizing constant. 
Then, we have:
\begin{align}
    \int_{a_k}^{b_k} \phi\left(\frac{x - p}{\sigma}\right) dx\phi(p) \triangleq S_{k,\sigma}(p) = \left[\Phi\left(\frac{b_k - p}{\sigma} \right) -\Phi\left(\frac{a_k - p}{\sigma} \right) \right]\phi(p)
\end{align}

and the probability that a codeword bit is flipped due to a transform against which $f$ is $\sigma$-robust is:

\begin{equation}
     p_{(\Delta,\delta)}(\sigma) =  1 - \left( 2 \sum_{j=\mathbb{Z}} \sum_{k=\mathbb{Z}}\frac{1}{Z_j}  \int_{\alpha_j}^{\beta_j} S_{k,\sigma}(p) dp \right) ; S_{k,\sigma}(p) = \left[\Phi\left(\frac{b_k - p}{\sigma} \right) -\Phi\left(\frac{a_k - p}{\sigma} \right) \right]\phi(p)
\end{equation}

\end{proof}

\subsection{Proposition~\ref{prop:ssb-fidelity}: \method Asymptotic fidelity}\label{app:asymp-dist}

\begin{appendixprop}[\method Asymptotic relative fidelity]\label{prop:app-ssb-fidelity}
    Any $(\Delta,\delta)$-SSB system with codeword size $M^\prime$ is $M^{\prime}\zeta(\Delta, \delta)$-faithful with 
    \begin{equation}
            \zeta(\Delta, \delta) \xrightarrow[]{L\rightarrow \infty} \frac{1}{2}\left[\left(1 + \mu_{(\Delta,\delta)} ^2 \right)\sigma_{(\Delta,\delta)} ^{-2} + \log \sigma_{(\Delta,\delta)} ^2 - 1\right],
    \end{equation}
    if $\bU$ is chosen such that Lyapunov condition holds.
\end{appendixprop}
\begin{proof}
    This is a direct consequence of Proposition \ref{app:asymp-dist} applied to Def.\ref{def:fidelity} and observing that a standard Gaussian distribution is still a standard Gaussian in watermarked space since columns of $\bU$ sum to 1.
\end{proof}

\subsection{Proposition~\ref{proposition:ssb-sec}}\label{app:proof-MP-sec}

\begin{appendixprop}[\method Watermarking security]
 Let $\alpha = \frac{M^\prime}{L}$. Any $(\Delta,\delta)$-SSB is $\eta$-secure against a PCA estimator where $\eta = \left( \frac{1- \sqrt{\alpha}\sigma_{\Delta,\delta}}{1 - \sigma_{\Delta,\delta}} \right)^2$.
\end{appendixprop}

\begin{proof}
\method's security relies on eigenvalues of its covariance matrix, and if the number of observations is enough to observe a shift from the expected Marchenko-Pastur distribution. 
First, we compute the covariance matrix $\Sigma_\bz$.
\begin{align}
    \Sigma_\bz &= \mathbb{E}[\bz \bz^\top] \\
    &= \E[\left((I_L - \bU\bU^\top)\bz + \bU\bz_\bu \right) \left( (I_L - \bU\bU^\top)\bz + \bU\bz_\bu\right)^\top] \\
    &= \E[\bU\bz_\bu {\bz_\bu}^\top\bU^\top] + \E[\bU{\bz_\bu}\bz_o^\top] \E[\bz_o{\bz_\bu}^\top\bU^\top] + \E[\bz_o\bz_o^\top]\\
\end{align}
with $\bz_o = (I_L - \bU\bU^\top)\bz$.
$\bU\bz_\bu$ and $\bz_o$ are independent, so $\E[\bU\bz_\bu\bz_o^\top] = \E[\bz_o\bz_\bu^\top\bU^\top] = 0$. 
Let's compute $\E[\bz_o\bz_o^\top]$.
\begin{align}
    \E[\bz_o\bz_o^\top] &= \E[(I_L - \bU\bU^T)\bz \bz^\top(I_L - \bU\bU^T)] \\
    &= (I_L - \bU\bU^T)\E[\bz\bz^\top](I_L - \bU\bU^T) \\
    &= (I_L - \bU\bU^T)(I_L - \bU\bU^T) \\
    &= I_L - \bU\bU^T. \\
\end{align}

It remains to calculate $\E[\bU\bz_\bu\bz_\bu^\top\bU^\top]$.
\begin{align}
    \E[\bU\bz_\bu\bz_\bu^\top\bU^\top] = \bU \E[\bz_\bu \bz_\bu^\top]\bU^\top,
\end{align}
where $\E[\bz_\bu \bz_\bu^\top]$ is the covariance matrix of $\bz_\bu$.
As $\bz_\bu$ are independent, $\text{Cov}(\bz_\bu) = \V(\bz_\bu) I_{M^\prime}$.

So, $\Sigma_\bz$ is given by:
\begin{align}
    \Sigma_\bz &= \bU \E[\bz_\bu \bz_\bu^\top]\bU^\top + \E[\bz_o\bz_o^\top]\\ 
    &= I_L + \left(\V(\bz_\bu) - 1\right)\bU\bU^\top
\end{align}

So eigenvalues are $\lambda_0 = 1$ for $(L - M^\prime)$ dimensions and $\lambda_1 = \V(\bz_\bu)$ for $M^\prime$ dimensions.
In a case of Gaussian sampling such that $\bz \sim \mathcal{N}(0, I_L)$, the eigenvalues follows a Marchenko-Pastur distribution with support $S = \left[\left(1 - \sqrt{\frac{L}{N}} \right)^2, \left(1 + \sqrt{\frac{L}{N}} \right)^2 \right]$, with $N$ the number of observations.
Following the Marchenko-Pastur theory, the support of eigenvalues in watermarking sampling is $S_{\lambda_0} = \left[\left(1 - \sqrt{\frac{L - M^\prime}{N}} \right)^2, \left(1 + \sqrt{\frac{L - M^\prime}{N}} \right)^2 \right]$ for $(L - M^\prime)$ dimensions and $S_{\lambda_1} = \left[\lambda_1\left(1 - \sqrt{\frac{M^\prime}{N}} \right)^2, \lambda_1\left(1 + \sqrt{\frac{M^\prime}{N}} \right)^2 \right]$ for $M^\prime$ dimensions.
A watermarking scheme is $\eta$-secure if $S_{\lambda_1} \subseteq S$. 
It gives the following condition:

\textbf{If $\lambda_1 > 1$:}
\begin{align}
    \lambda_1\left(1 + \sqrt{\frac{M^\prime}{\eta}}\right)^2 &= \left(1 + \sqrt{\frac{L}{\eta}} \right)^2 \\
    \sqrt{\lambda_1}\left(1 + \sqrt{\frac{M^\prime}{\eta}}\right) &= \left(1 + \sqrt{\frac{L}{\eta}} \right) \\
    \sqrt{\lambda_1} + \frac{\sqrt{\lambda_1 M^\prime}}{\sqrt{\eta}} &= 1 + \frac{\sqrt{L}}{\sqrt{\eta}} \\
    \sqrt{\lambda_1} - 1 &= \frac{\sqrt{L} - \sqrt{\lambda_1 M^\prime}}{\sqrt{\eta}} \\
    \sqrt{\eta}(\sqrt{\lambda_1} - 1) &= \sqrt{L} - \sqrt{\lambda_1 M^\prime} \\
    \sqrt{\eta} &= \frac{\sqrt{L} - \sqrt{\lambda_1 M^\prime}}{\sqrt{\lambda_1} - 1} \\
    \eta = \left( \frac{1- \sqrt{\alpha}\sigma_u}{1 - \sigma_u} \right)^2,
\end{align}
with $\alpha = \frac{M^\prime}{L}$ and $\sigma_u = \sqrt{\lambda_1}$.

\textbf{If $\lambda_1 \leq 1$:}
\begin{align}
    \lambda_1\left(1 - \sqrt{\frac{M^\prime}{\eta}}\right)^2 &= \left(1 - \sqrt{\frac{L}{\eta}} \right)^2 \\
\end{align}
so the proof is analogous.

Then, we distinguish the two extreme cases for the computation of the eigenvalues. 
We know that $\lambda_1 = \V(\bz_\bu)$.

\begin{align}
    \E[\bz_\bu] &= 2\mathrm{sign}(c)\sum_{k=-\infty}^{+\infty}(\Phi(b_k) - \Phi(a_k)) \frac{\int_{\alpha_k}^{\beta_k}x\phi(x)dx}{\Phi(\beta_k) - \Phi(\alpha_k)} \\
    &= 2\mathrm{sign}(c)\sum_{k=-\infty}^{+\infty}(\Phi(b_k) - \Phi(a_k)) \frac{\phi(\alpha_k) - \phi(\beta_k)}{\Phi(\beta_k) - \Phi(\alpha_k)} \\
    &= \mathrm{sign}(c)\sum_{k=-\infty}^{+\infty}P_k \frac{\phi(\alpha_k) - \phi(\beta_k)}{\Phi(\beta_k) - \Phi(\alpha_k)}, \\
\end{align}
with $P_k = 2(\Phi(b_k) - \Phi(a_k))$.

and, 

\begin{align}
    \E[\bz_\bu^2] &= 2\sum_{k=-\infty}^{+\infty}(\Phi(b_k) - \Phi(a_k)) \frac{\int_{\alpha_k}^{\beta_k}x^2\phi(x)dx}{\Phi(\beta_k) - \Phi(\alpha_k)} \\
    &= 2\sum_{k=-\infty}^{+\infty}(\Phi(b_k) - \Phi(a_k)) \frac{\Phi(\beta_k) - \beta_k \phi(\beta_k) - \Phi(\alpha_k) + \alpha_k\phi(\alpha_k)}{\Phi(\beta_k) - \Phi(\alpha_k)} \\
    &= \sum_{k=-\infty}^{+\infty}P_k \left(1-\frac{\beta_k\phi(\beta_k) - \alpha\phi(\alpha_k)}{\Phi(\beta_k) - \Phi(\alpha_k)}\right), \\
\end{align}
with $P_k = 2(\Phi(b_k) - \Phi(a_k))$.

Then, the covariance is given by:
\begin{align}
    \mathrm{Cov}(\bz_\bu) &= I_L\V(\bz_\bu)\\
    &= I_l \left(\E[\bz_\bu^2] - \E[\bz_\bu]^2 \right) \\
    &= \mathbf{I}_L\left(\sum_{k=-\infty}^{+\infty} P_k \left(1-\frac{\beta_k\phi(\beta_k) - \alpha\phi(\alpha_k)}{\Phi(\beta_k) - \Phi(\alpha_k)}\right) - \bmu_u^2 \right),
\end{align}
with $\bmu_u^2 = \E[\bz_\bu]^2$.
\end{proof}



\section{Perfect security of \method}\label{app:perfect-security}
    \begin{figure}[h!]
    \centering
    \includegraphics[width=0.65\linewidth]{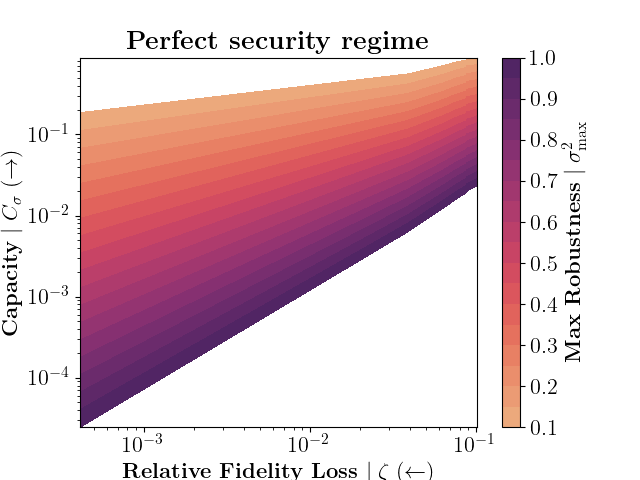}
    \caption{Perfect security regime with $(1.6, 0)$-SSB.}
    \label{fig:perfect-sec}
    \end{figure}
An interesting behavior of \method can be observed for the security ratio in Figure~\ref{fig:ssb-fullcharac} around $\Delta = 1.6$ when using the finest lattice ($\delta \rightarrow 0$). The observed peak is due to Proposition~\ref{proposition:ssb-sec}, which states that if the eigenvalues of the watermarked covariance matrix are all $1$, then one cannot estimate the secret rotation $\bU$, whatever how many samples we provide to Eve. A surprising fact is that \method can achieve perfect security for non-zero capacities. Even better, for any $\Delta < 1.6$, one can find a corresponding $\delta > 0$ for which it achieves perfect security.
One further advantage of the perfect security regime is that it can be studied independently of the message size $M$ and of the latent space dimension $L$: both fidelity and capacity are computed per latent element. We plot the watermarking characteristic under this regime in Figure~\ref{fig:perfect-sec}.


\section{Examples}
Figure~\ref{fig:images_sana} illustrates examples of images generated with Sana with the compared multi-bit seed-based watermarking methods. 

\begin{figure}[t]
    \centering
    \setlength{\tabcolsep}{1pt}  
    \renewcommand{\arraystretch}{0}
    \begin{tabular}{ccccc}
        \includegraphics[width=0.18\textwidth]{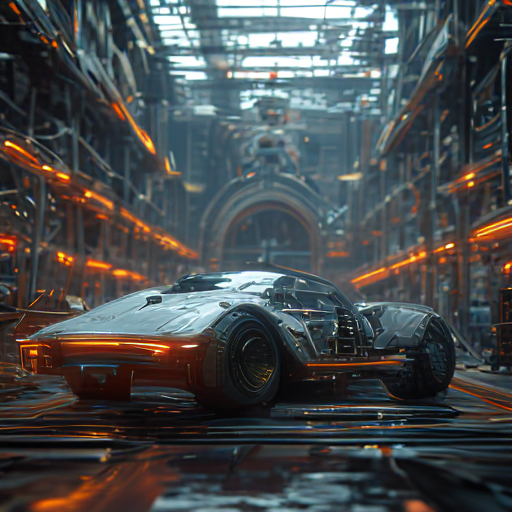} &
        \includegraphics[width=0.18\textwidth]{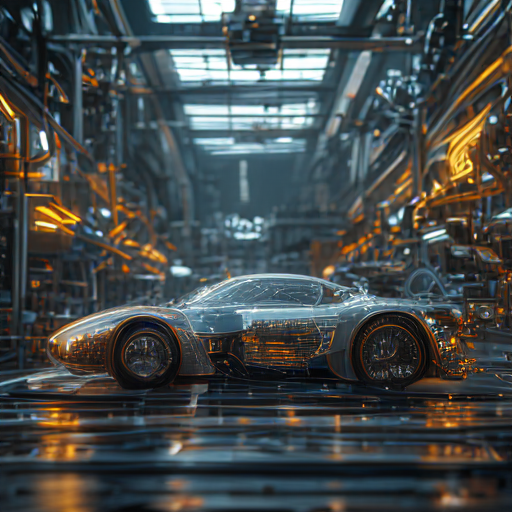} &
        \includegraphics[width=0.18\textwidth]{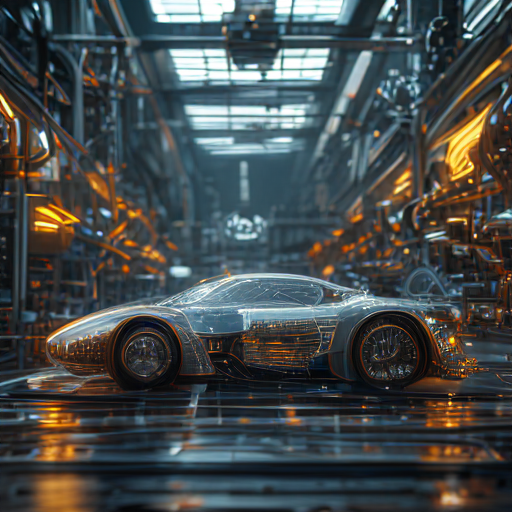} &
        \includegraphics[width=0.18\textwidth]{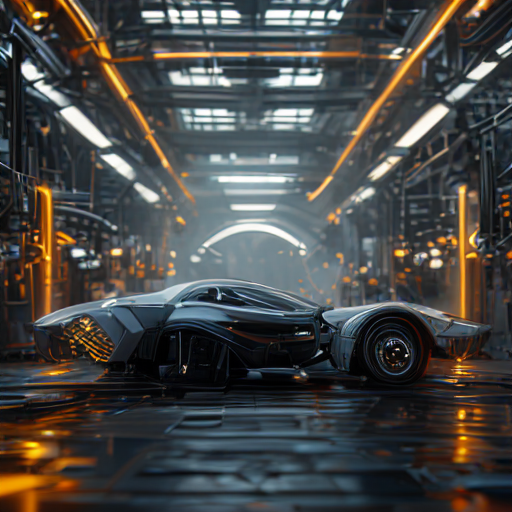} &
        \includegraphics[width=0.18\textwidth]{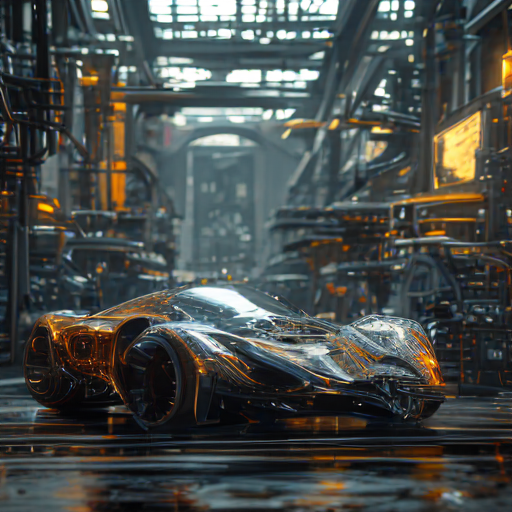} \\
        \includegraphics[width=0.18\textwidth]{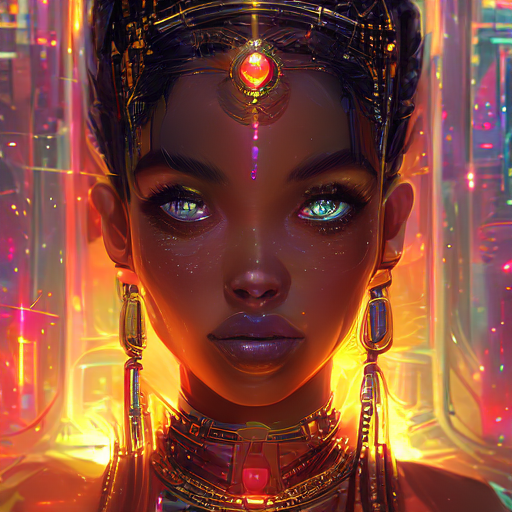} &
        \includegraphics[width=0.18\textwidth]{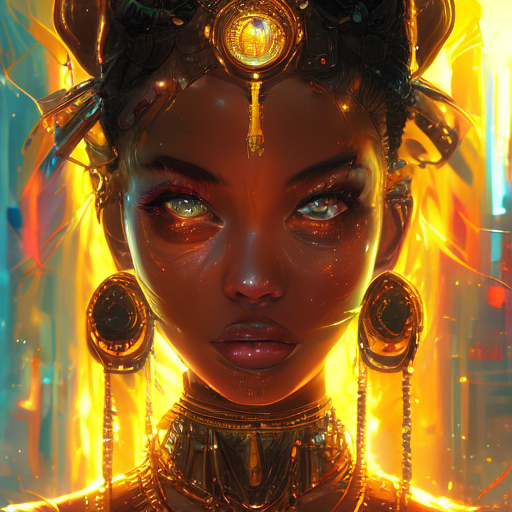} &
        \includegraphics[width=0.18\textwidth]{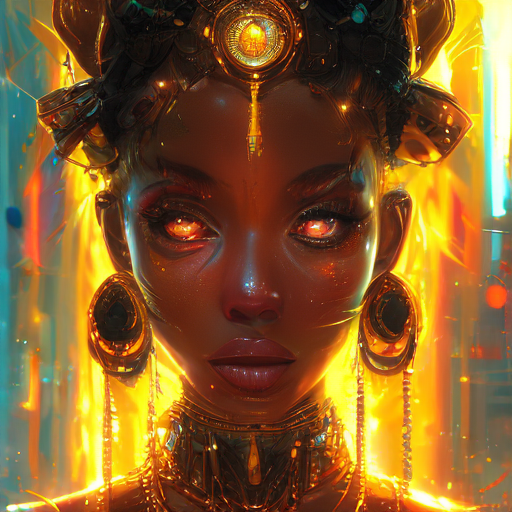} &
        \includegraphics[width=0.18\textwidth]{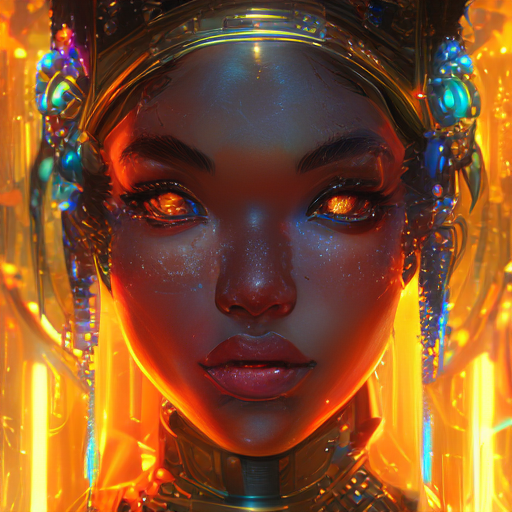} &
        \includegraphics[width=0.18\textwidth]{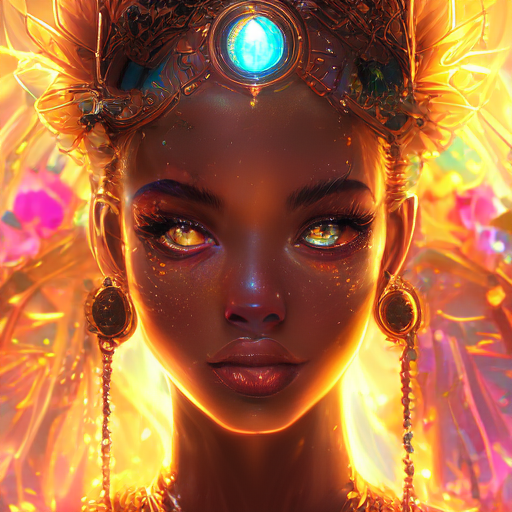} \\
        \includegraphics[width=0.18\textwidth]{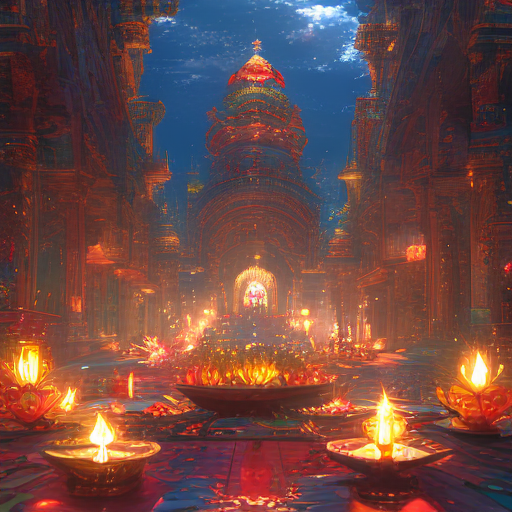} &
        \includegraphics[width=0.18\textwidth]{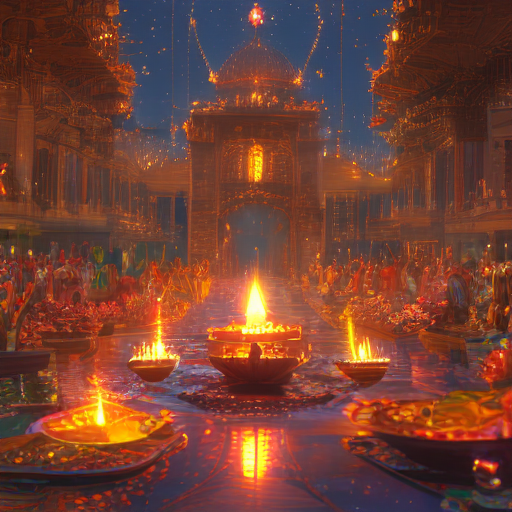} &
        \includegraphics[width=0.18\textwidth]{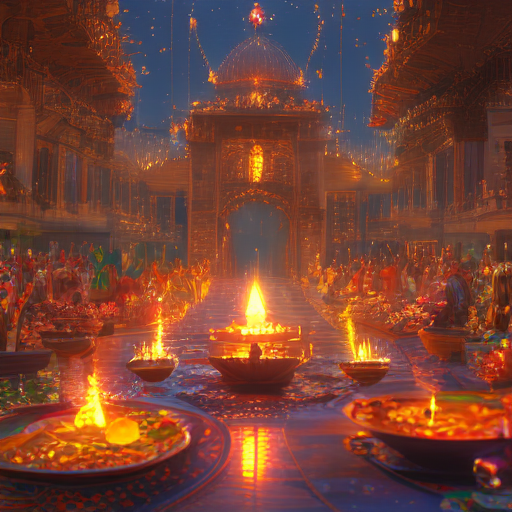} &
        \includegraphics[width=0.18\textwidth]{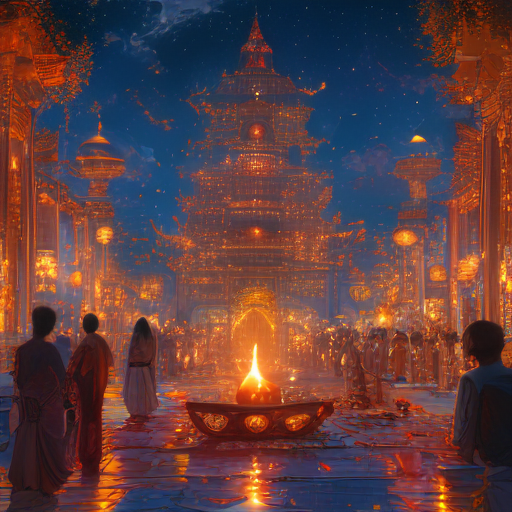} &
        \includegraphics[width=0.18\textwidth]{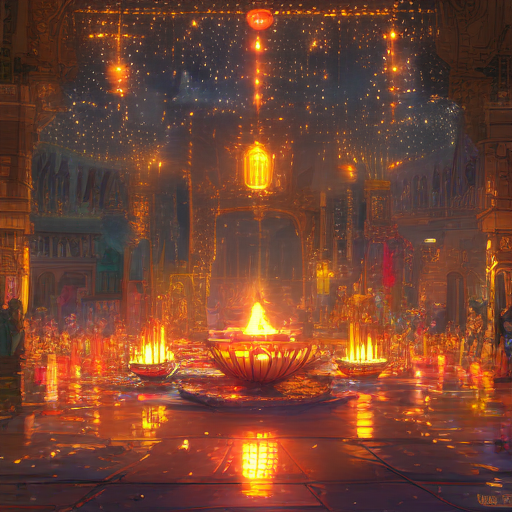} \\
        \includegraphics[width=0.18\textwidth]{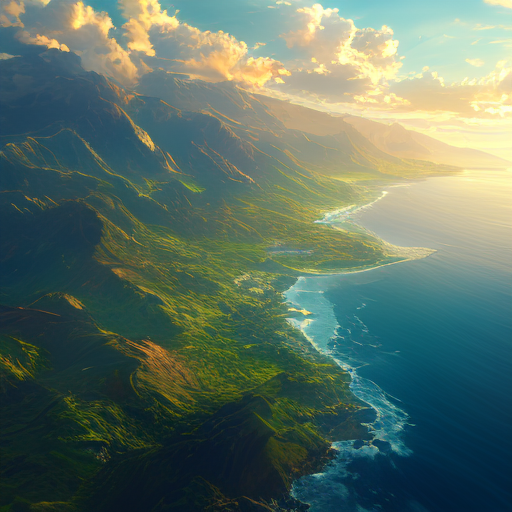} &
        \includegraphics[width=0.18\textwidth]{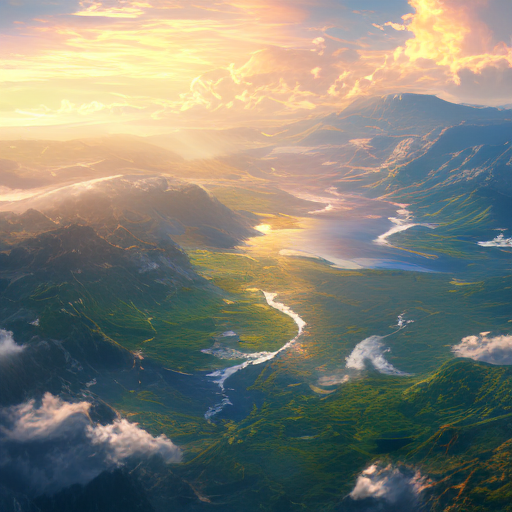} &
        \includegraphics[width=0.18\textwidth]{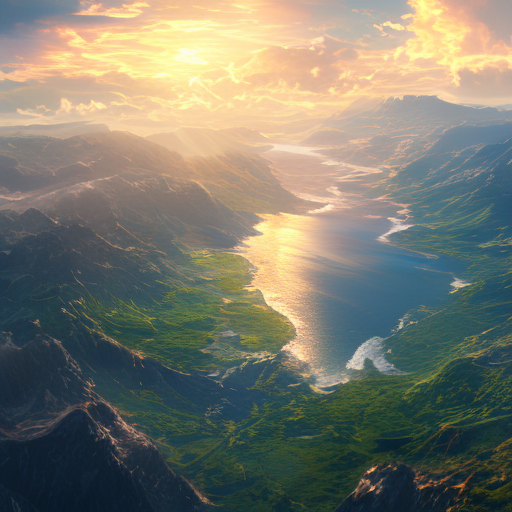} &
        \includegraphics[width=0.18\textwidth]{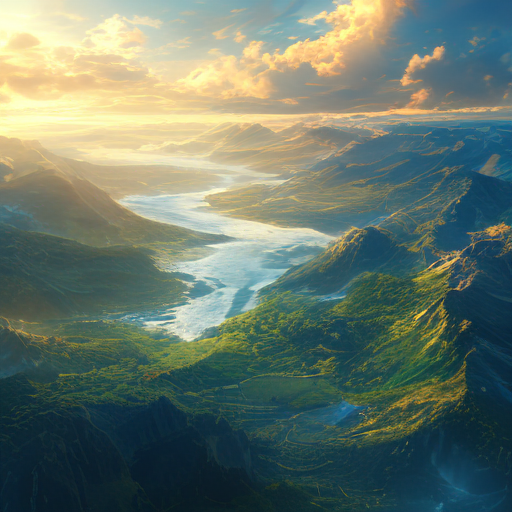} &
        \includegraphics[width=0.18\textwidth]{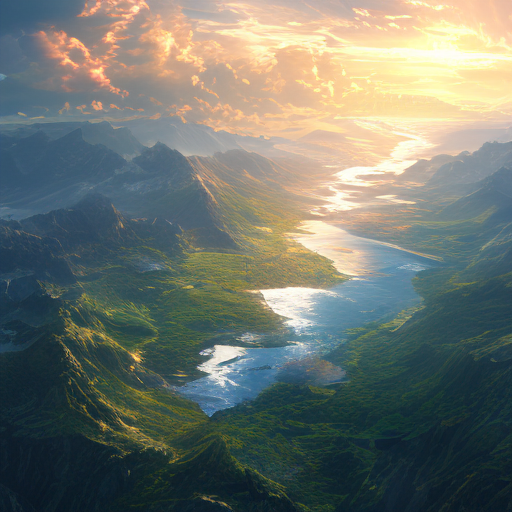} \\
        \includegraphics[width=0.18\textwidth]{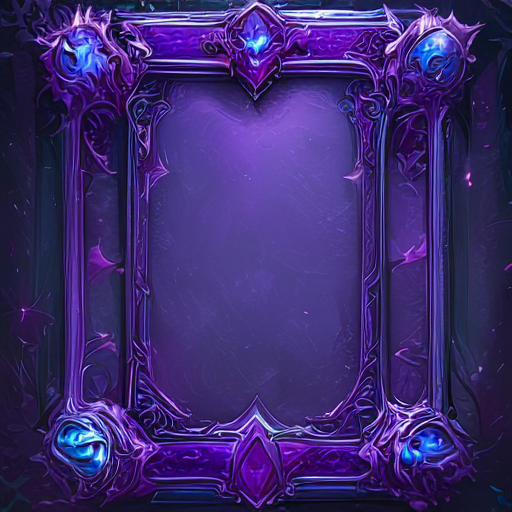} &
        \includegraphics[width=0.18\textwidth]{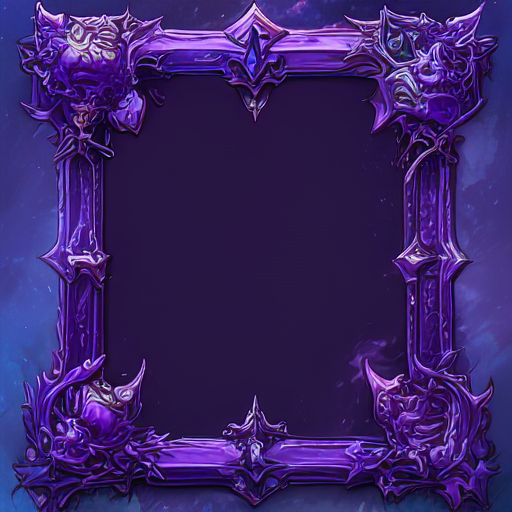} &
        \includegraphics[width=0.18\textwidth]{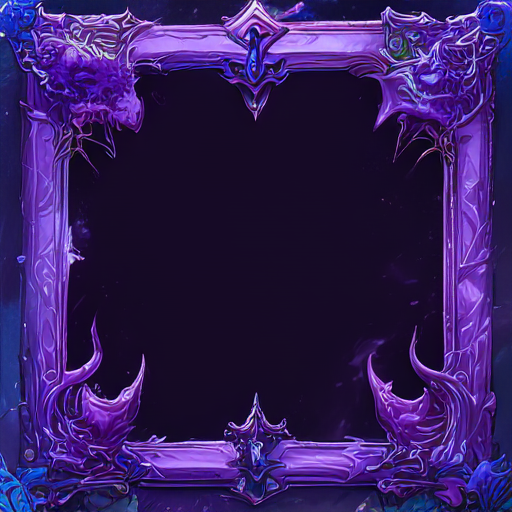} &
        \includegraphics[width=0.18\textwidth]{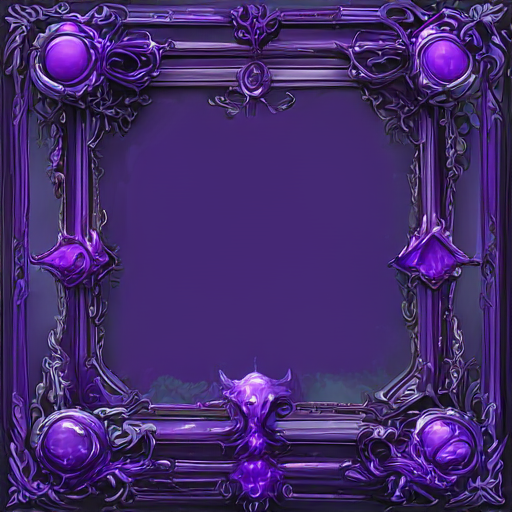} &
        \includegraphics[width=0.18\textwidth]{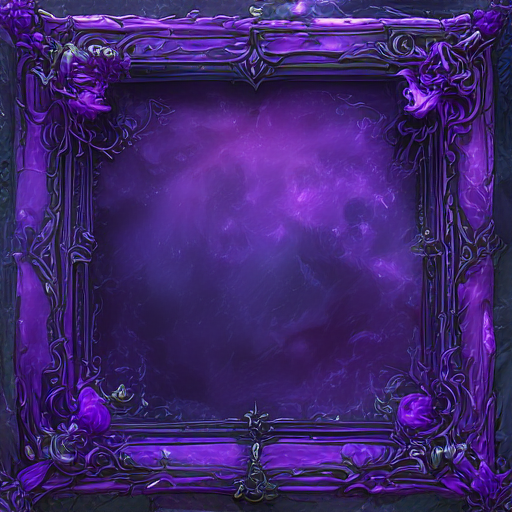} \\[4pt]
        No watermark & $(1.6,0)$-SSB & $(+\infty,+\infty)$-SSB & Gaussian-Shading & PRC \\
    \end{tabular}
    \caption{Example of images generated with Sana. Images were selected randomly.}
    \label{fig:images_sana}
\end{figure}

\section{Reproducibility Statement}\label{app:reproductibility}

All models and datasets are detailed in Appendix~\ref{app:licences}. 
All the simulations can be executed in the notebook provided as a supplementary file. 
The full code used for the experimental validation will be released upon acceptance. 
The implementation relies on PyTorch and Diffusers libraries, with fixed random seeds.
All hyperparameters are provided in Section~\ref{sec:experiments} to achieve the results presented in this work.
Experiments were conducted on NVIDIA L40s GPUs for 20 hours.  


\section{Licenses}
\label{app:licences}

\subsection*{Datasets}

\textbf{Stable-Diffusion-Prompts}
\begin{itemize}
    \item Source: \url{https://huggingface.co/datasets/Gustavosta/Stable-Diffusion-Prompts}
    \item License: Unknown
\end{itemize}

\subsection*{Pretrained Diffusion Models}

\textbf{Sana}

\begin{itemize}
    \item Source: \url{https://huggingface.co/Efficient-Large-Model/Sana_600M_512px}
    \item License: NVIDIA License
\end{itemize}

\textbf{Z-image}

\begin{itemize}
    \item Source: \url{https://huggingface.co/Tongyi-MAI/Z-Image}
    \item License: apache-2.0
\end{itemize}

\textbf{Qwen-Image}

\begin{itemize}
    \item Source: \url{https://huggingface.co/Qwen/Qwen-Image}
    \item License: apache-2.0
\end{itemize}

\subsection*{Augmentations}

\textbf{Augly}
\begin{itemize}
    \item Source: \url{https://github.com/facebookresearch/AugLy}
    \item License: MIT License
\end{itemize}

\end{document}